\documentclass{article}
\usepackage[utf8]{inputenc}
\usepackage{amsmath}
\usepackage{amssymb}
\usepackage{float}
\usepackage{graphicx,graphbox}
\usepackage{amsthm}
\usepackage{multicol}
\usepackage{multirow}
\usepackage{comment}
\usepackage{enumerate}
\usepackage{xparse}
\usepackage[margin=1in]{geometry} 
\usepackage{algorithm,algpseudocode,float}

\usepackage{subfig}
\usepackage{hyperref}%
\usepackage{dirtytalk}

\hypersetup{
    colorlinks=true,
    linkcolor= black,
    urlcolor= black,
    citecolor = black
}
\usepackage{hyperref}
\usepackage{wrapfig}
\usepackage{enumitem}

\usepackage[capitalize]{cleveref}[2012/02/15]%

\graphicspath{{./figures/}}

\crefformat{footnote}{#2\footnotemark[#1]#3}

\makeatletter
\newenvironment{breakablealgorithm}
  {%
   \begin{center}
     \refstepcounter{algorithm}%
     \hrule height.8pt depth0pt \kern2pt%
     \renewcommand{\caption}[2][\relax]{%
       {\raggedright\textbf{\ALG@name~\thealgorithm} ##2\par}%
       \ifx\relax##1\relax %
         \addcontentsline{loa}{algorithm}{\protect\numberline{\thealgorithm}##2}%
       \else %
         \addcontentsline{loa}{algorithm}{\protect\numberline{\thealgorithm}##1}%
       \fi
       \kern2pt\hrule\kern2pt
     }
  }{%
     \kern2pt\hrule\relax%
   \end{center}
  }
\makeatother

\usepackage[backend=biber,bibencoding=utf8,  %
			url=false,sorting=ynt,style=numeric-comp,  %
			maxbibnames = 99]{biblatex}                 %
\newtheorem{theorem}{Theorem}[section]
\newtheorem{defn}[theorem]{Definition}
\newtheorem{corollary}[theorem]{Corollary}
\newtheorem{lemma}[theorem]{Lemma}

\usepackage[dvipsnames]{xcolor} 

\newcommand{\R}{\mathbb{R}}
\newcommand{\mc}{\mathrm{mc }}
\newcommand{\rc}{\mathrm{rc }}
\newcommand{\isepssimilar}{\texttt{IS}E\texttt{PSSIMILAR} }
\newcommand{\cX}{\mathcal{X}}
\newcommand{\inv}{^{-1}}

\newcommand{\e}{\varepsilon}
\renewcommand{\epsilon}{\varepsilon}

\begin{document}
\begin{center}
\begin{Large}\textbf{{Comparing Embedded Graphs Using Average Branching Distance}}

\begin{normalsize}
 \begin{multicols}{3}

 \textbf{Levent Batakci} 
 
 Case Western Reserve University
 
 \textit{lab192@case.edu }
 \vspace{1em}

 \textbf{Abigail Branson} 
  
  Union University
  
 \textit{abby.branson@my.uu.edu}
  \vspace{1em}

 \textbf{Bryan Castillo}
  
  Mesa Community College
  
 \textit{bryancastillo98@gmail.com}
  \vspace{1em}

 \textbf{Candace Todd} 
  \vspace{0.25em}
  
  The Pennsylvania State University
  \vspace{0.25em}
  
 \textit{clt5441@psu.edu}
  \vspace{1em}
  
  \columnbreak

  \textbf{Erin Wolf Chambers}
  
  Saint Louis University
  
 \textit{erin.chambers@slu.edu }
  \vspace{1em}
 
 \textbf{Elizabeth Munch} 
  \vspace{0.25em}
  
  Michigan State University
  \vspace{0.25em}
  
 \textit{muncheli@msu.edu}
  \vspace{1em}

 \end{multicols}
 \end{normalsize}

\end{Large}

\end{center}

\begin{abstract}
Graphs drawn in the plane are ubiquitous, arising from data sets through a variety of methods ranging from GIS analysis to image classification to shape analysis. A fundamental problem in this type of data is comparison: given a set of such graphs, can we rank how similar they are, in such a way that we capture their geometric ``shape" in the plane?

In this paper we explore a method to compare two such embedded graphs, via a simplified combinatorial representation called a  \textit{tail-less merge tree} which encodes the structure based on a fixed direction.
First, we examine the properties of a distance designed to compare merge trees called the \textit{branching distance}, and show that the distance as defined in previous work fails to satisfy some of the requirements of a metric.
We  incorporate this into a new distance function called \textit{average branching distance} to compare graphs by looking at the branching distance for merge trees defined over many directions. 
Despite the theoretical issues, we show that the definition is still quite useful in practice by using our open-source code to cluster data sets of embedded graphs.

\end{abstract}

\section{Introduction}

Embedded graphs appear in a wide variety of applications, including map reconstructions, GIS data, shape analysis, and medical imaging.  Such graphs are more than simple abstract representations, since they have geometric information attached, usually in the form of coordinates and edge lengths.
When faced with data of this type, an immediate question to ask is then how to compare them, so that we can apply techniques like clustering, image recognition, or machine learning. 
In essence, we seek a \textit{distance measure}, which given two input graphs, returns a number representing how similar they are.  Ideally, this  value should be zero if they are the same, and take increasing values for increasingly different embedded graphs.  

While solving this problem perfectly is of course difficult in practice, many examples of graph distances do exist in the literature. 
For example, one well-known one is the graph edit distance \cite{Gao2009}, which gives a cost to inserting or removing vertices and edges; the distance is then the minimum possible sequence of edge and vertex operations when transforming one graph into the other.  
Unfortunately, this particular measure does not take any embedding information into account, so two graphs with identical vertices and weights which are embedded in a very different configuration will still have zero distance.

In this paper, we develop and implement a technique to  compare 2-dimensional embedded graphs, such as are generated from GIS data or skeletonization of images.
We note that there are available options which are close to our work, where distances  are specifically built to take the embedding into account \cite{ahmed2014local,cheong2009measuring,biagioni2012inferring,Alt2003,karagiorgou2012vehicle,buchin2017distance}.  
However, we seek to develop a faster pipeline which incorporates graph simplification while still retaining at least part of the geometric embedding information.
More specifically, we compare embedded graphs by replacing them with a simpler graph which still encodes some aspect of the structure. 

Our algorithm fixes a direction in the plane and computes the \textit{merge tree} of the graph, which is a tree with a real valued function representing how the connected components of sub-level sets of the graph changes in that direction; see e.g.~\cref{fig:merge_components}.
To be able to compare merge trees in a meaningful way, we need a distance function that is computable, accurate, comprehensive, and versatile. 
There are many possible options for metrics to compare merge trees specifically, including 
the edit distance \cite{Sridharamurthy2018}, and
interleaving distance \cite{Morozov2013,Munch2019,Gasparovic2019,Yan2019a}.
In this paper, we use the branching distance for merge trees \cite{morozov}.
Then, to avoid bias caused by fixing a single a direction, we rotate the graphs, calculating the merge trees at each rotation, then compare all of the resulting merge trees.
We take the median of these distances, naming this distance as \emph{average branching distance}. 

The branching distance for merge trees \cite{morozov} builds on the idea of a branch decomposition for a broader class of trees, introduced in \cite{Pascucci2004a}.
However, in the course of our work, we have discovered an issue with the definition presented in \cite{morozov}, caused by ignoring the infinite tail present in merge trees. 
The result is that the distance as defined does not satisfy the properties of a metric, and we provide counterexamples to that effect.
We then prove that this issue propogates into the average branching distance, so it also is not a metric, semi-metric, or pseudo-metric.
We also explore its theoretical properties on simple classes of graphs such as convex polygons.

Despite this, we show that the average branch decomposition is still a useful tool for distinguishing embedded graphs in practice. 
To test its utility on practical examples, we tested it with visualizations such as dendrograms and lower dimensional embeddings, using a range of available data sets. 
Our code is posted publicly~\cite{us}.
We found that despite the theoretical issues, the distance still is potentially useful for understanding real data.

\section{Background}
In this section, we give the relevant background terminology and results in graph theory and merge trees, culminating in the definition of the branching distance for merge trees (Defn.~\ref{defn:branchingDistance}).

\subsection{Graph theory}

Here, we provide a brief discussion of relevant terms from traditional graph theory, following \cite{bondymurty}. 

A \textit{graph} $G = (V,E)$ consists of a set of vertices $V = V(G)$,  and a set of unordered pairs of vertices called edges $E = E(G)$.
An edge $e = uv$ is \textit{incident} to either of its endpoints $u$ and $v$. 
An \textit{edge deletion} operation removes an edge from $E(G)$.
A \textit{vertex deletion} operation removes a vertex $v$ from $V(G)$ and removes all edges incident to $v$ from $E(G)$

A \textit{loop} is an edge such that both endpoints are the same vertex. 
Two edges are called \textit{parallel} if their endpoints share the same vertices. 
A \textit{simple graph} has no loops or parallel edges. 
The \textit{degree} of a vertex $v$ is the number of edges incident to $v$, with loops counting as two edges. 

A \textit{subgraph} $H$ of a graph $G$, denoted $H \subseteq G$, is a graph with $V(H) \subseteq V(G)$ and $E(H) \subseteq E(H)$.
Note that a subgraph of $G$ can be obtained  by performing a combination of vertex deletions and/or edge deletions on $G$. 
An \textit{induced subgraph} of $G$ is a subgraph obtained solely by vertex deletions in $G$.
In this case, we say that $H$ is a \textit{subgraph induced by $A$} if $V(H) = A \subseteq V(G)$.
Two graphs $G$ and $H$ are \textit{isomorphic} when there exists a bijective mapping between the vertices of $G$ and $H$ that preserves adjacency.

A \textit{path} in a graph is a sequence of ordered vertices $v_1,\cdots,v_n$ with edges $v_i v_{i+1}$ for $i = 1,\cdots,n-1$.
A graph is \textit{connected} if every pair of vertices can be connected by a path. 
A \textit{connected component} of a graph $G$ is a connected subgraph $H \subseteq G$ which is not contained in a larger connected subgraph. 
A \textit{cycle} is a path whose start and end vertices are the same vertex. 
A graph is \textit{acyclic} if it has no subgraph which is a cycle graph.
A \textit{tree} is a connected acyclic graph, and a \textit{forest} is a graph solely consisting of trees.

We assume we have a \textit{topological graph}, where each edge can be considered as a copy of the unit interval.
Let $G$ be a topological graph and let $f$ be a function on $G$,  $f: G \rightarrow \mathbb{R}$, providing every vertex and point on each edge of the graph with a placement on the real number line. 
See an example in \cref{fig:merge_components}.
For a fixed $a \in \mathbb{R}$, the \textit{sub-level set at $a$} is an induced subgraph $H$ of $G$ such that any vertex with a function value in $(a,\infty)$ is deleted; that is, it is the subgraph induced  by the set of vertices with function value $f(v) \leq a$.
In this case, we write $H=f^{-1}((-\infty,a])$.
The \textit{sublevel set below $a$} is the subgraph induced by the set of vertices with function value $f(v) < a$. 
A component is a set of connected vertices in the sub-level set $A$.

\begin{figure}
        \centering
        \includegraphics[width = \textwidth]{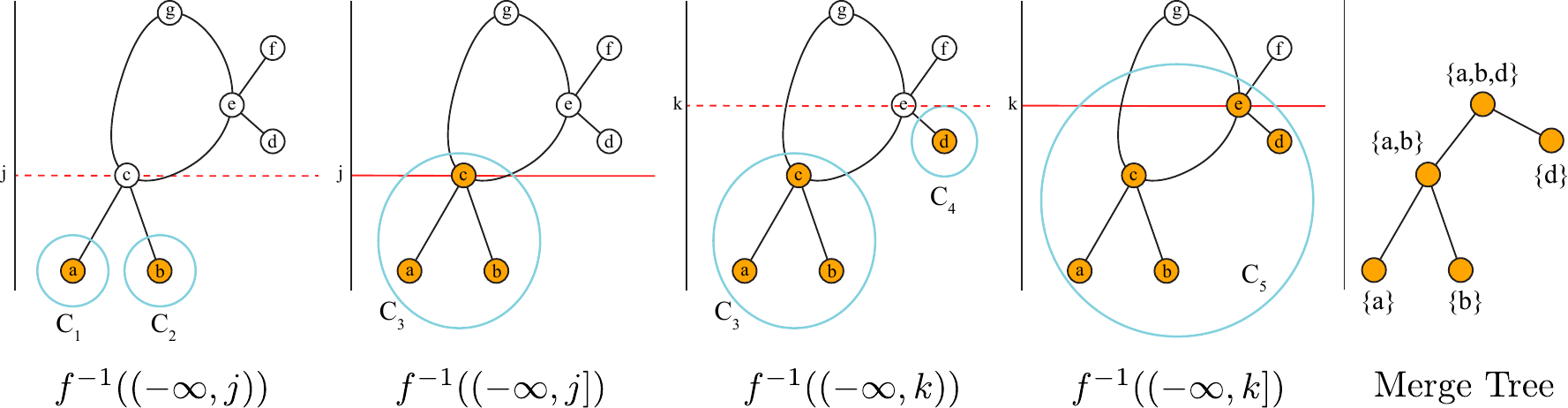}
    \caption{The sublevel sets of a graph below and at two values, $j,k \in \R$; and the tail-lessmerge tree of the graph.
    }%
    \label{fig:merge_components}%
\end{figure}

\subsection{Merge Trees}

A merge tree is a representation of the connected components of sublevel sets of a graph with an $\mathbb{R}$-valued function. 
The vertices of a merge tree represent changes in connectedness of the input graph and its edges encode the relationships between these changes.
In this way, the merge tree serves as a less complex representation of the graph.
What follows is a formal definition of a merge tree closely related to the dendrogram definition of \cite{Carlsson2010a}.

A merge tree $M$ is a structure defined on a graph $G$ with a function $f: G \rightarrow \mathbb{R}$ which encodes the changing connected components of $f \inv(-\infty,a]$ for a shifting value $a$.
For simplicity, we assume the function is entirely determined by values on the vertices $f:V(G) \rightarrow \mathbb{R}$ given by $v \mapsto f(v_i)$ and extended linearly to the edges.
We use the following notation to define merge trees.
Fix a connected set $A \subseteq \mathbb{R}$.
Let $f^{-1}(A)$ denote the induced subgraph of $G$ containing just the vertices $v \in G$ such that  $f(v) \in A$.
Fix a \textit{connected subgraph} $C \subseteq G$. 
Define the \textit{minima} $\mu(C)$ of  $C$ to be the set of vertices in $C$ having no neighbors of lesser function value; that is
    \begin{equation*}
        \mu(C) = \{ v \in V(C) \mid f(v) \leq f(w) \forall (vw) \in E(C) \} .
    \end{equation*}

Let the \textit{identified components} on $A$ be defined as $\Gamma(A) = \{\mu(C) \mid C \text{ is a connected component of } f^{-1}(A) \}$. 
Note that $\Gamma(A)$ is a set of sets and that its elements are in one to one correspondence with the connected components.  
Let the \textit{change in connectedness} $\Delta(a)$ at $a$ be defined as $\Gamma((-\infty, a]) \setminus \Gamma((-\infty,a))$. 
    Note that there is a change in connectedness in the sublevel sets of $G$ at $a$ if and only if $\Delta(a) \neq \emptyset$.

We note that we are defining merge trees as used in \cite{morozov}, which means that we inherited a bug in the definition from what they work with implicitly. 
Normally, a merge tree would have an infinite upwards tail representing the fact that a connected component is always visible in $f\inv(\infty,a]$ for $a$ greater than the maximum value. 
However, \cite{morozov} does not utilize this tail, cutting off the merge tree at the last merge. 
For this reason, we call this construction the \textit{tail-less} merge tree to distinguish it from more standard mathematical constructions.
For the remainder of the paper, we call this construction simply the ``merge tree'' unless there is  potential for confusion.

\begin{defn}[Tail-less Merge Tree]
\label{M} 
\label{defn:MergeTree}
The (tail-less) merge tree $M$ of a graph $G$ with function $f:G\rightarrow\mathbb{R}$ 
is a graph given by 
\begin{equation*}
\begin{split}
V(M) &= \{L \mid L \in \Delta(a), a \in \mathbb{R} \},\\
E(M) &= \{L_1L_2 \mid  L_1 \in \Delta(a) \text{ for } a \in \mathbb{R},  L_2 \in \Gamma((-\infty, a))\\
        & \qquad \qquad \qquad \text{ and } L_2 \subset L_1 \},\\
\end{split}
\end{equation*}
with function defined by
\begin{equation*}
\begin{array}{rccc}
f_M:& V(M) &\rightarrow &\mathbb{R}\\
&L &\mapsto & a \mid L \in \Delta(a).\\
\end{array}
\end{equation*}
\end{defn}

To understand these definitions, consider the example of  Fig.~\ref{fig:merge_components}. 
First, consider function value $j = f(c)$ for vertex $c$.
The highlighted portion  of the leftmost graph represents $f^{-1}((-\infty, j))$, which has only vertices $a$ and $b$. 
The two connected components are identified as $C_1$ and $C_2$.
Furthermore, $\mu(C_1)=\{a\}$ and $\mu(C_2)=\{b\}$.
Thus, $\Gamma((-\infty, j))=\{\{a\}, \{b\}\}$.
The second graph represents $f^{-1}((-\infty, j])$, where the only connected component is identified as $C_3$.
Here, $\mu(C_3)=\{a, b\}$ and $\Gamma((-\infty, j])=\{\{a, b\}\}$.
It follows that $\Delta(j) = \{\{a,b\}\}$.

Now we will compute $\Delta(k)$, where $k=f(e)$ for vertex $e$ by considering the third and fourth graphs in \cref{fig:merge_components}.
We see that $\mu(C_3)=\{a,b\}$ and $\mu(C_4)=\{d\}$.
Thus, $\Gamma((-\infty, k))=\{\{a,b\}, \{d\}\}$
Furthermore, $\mu(C_5)=\{a, b, d\}$ and $\Gamma((-\infty, k])=\{\{a, b,d\}\}$.
It follows that $\Delta(k) = \{\{a,b,d\}\}$.

In the far right of \cref{fig:merge_components}, we show the merge tree for the input graph and function. 
Each vertex $v$ of the merge tree with $f(v)=a$ is labeled with its corresponding element $L \in \Delta(a)$.
Edges are included between vertices $L_1$ and $L_2$ when $L_1 \subset L_2$.

\subsection{Merge tree for a fixed direction}
\label{ssec:DirectionTransformMergeTree}
\begin{figure}
    \centering
\includegraphics[width = \textwidth]{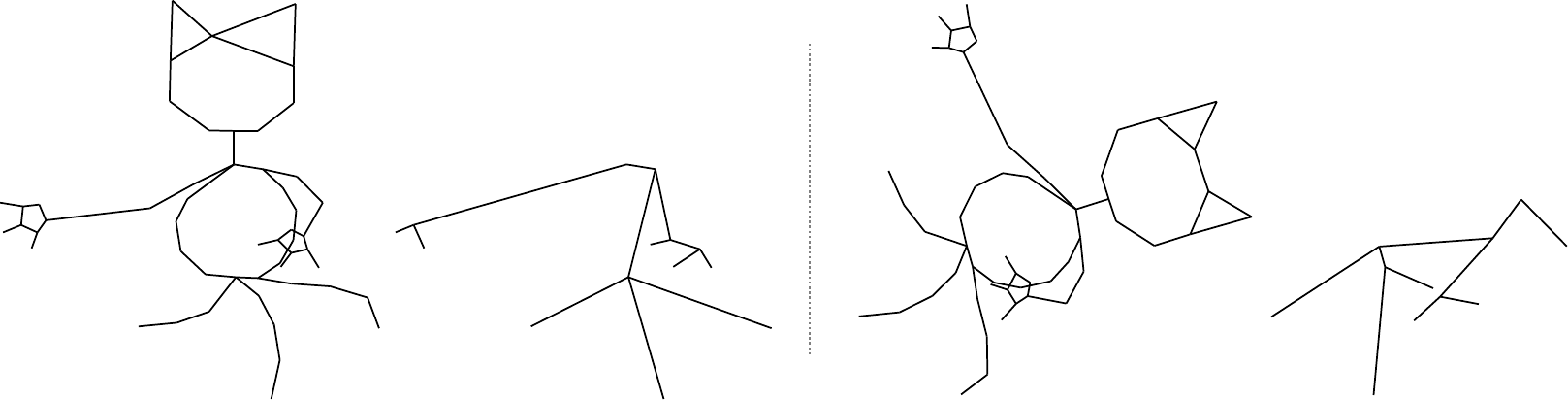}
    \caption{Two (nearly identical) embedded graphs at two different rotations with the associated merge trees.}
    \label{fig:RotatedCat}
\end{figure}
In this paper, we will consider graphs with a map to  $\mathbb{R}^2$. 
In the case that this map is injective, the result is an \textit{embedded graph}.%
\footnote{While we will often be interested in embedded graphs in practice, much of the work in this paper does not require the assumptions of injectivity.}
We can then pick a direction in the plane by fixing an angle $\omega$ for the orientation and calculating function values for a vertex as the magnitude of the projection of the original position vectors onto  $\omega$.
So, from the input $f:G \to \R^2$, we get a function $f_\omega:G \to \R$ for each angle $\omega \in [0,2\pi)$.
We normalize this function by shifting the function values to have median vertex value equal to 0. 
We get the same function if we think of rotating the graph in the plane and computing the height function in the vertical direction, i.e.~$f_\omega(v)$ is given by the $y$-coordinate of $f(v)$.
Then, we can compute the merge tree of the graph for each direction. 
See the example of Fig.~\ref{fig:RotatedCat}.

\subsection{Branching Distance}
In order to compare two merge trees, we use the following distance as defined by Beketayev et al.~\cite{morozov} which we call the \textit{branching distance}.
Note that this definition utilizes the tail-less merge tree construction as defined above.

In what follows, we assume the tail-less merge tree is a non-empty, connected graph.
The root of the merge tree is the vertex with highest function value. 
A merge tree is trivial if it consists of a single vertex. 
For a non-trivial, merge tree $M$, a \textit{saddle} is any vertex with degree $\geq 2$, and a \textit{minimum} is any non-root vertex with degree $1$. 
For a trivial merge tree, the only available vertex is considered to be both the root and a minimum.
A root branch is a pairing of some minimum $m_r$ with the highest function value vertex $s_r$ of the merge tree.

\begin{defn}[Branch Decomposition]

A branch decomposition $B$ of a tail-less merge tree $M$ is collection of pairs  $B = \{ (m_i, s_i)  \}_i$ consisting of a minima $m_i$ and either a saddle or the root $s_i$.
Every vertex appears in at least one pair.
Further, for each pair there is a descending path from the saddle $s_i$ to the minimum $m_i$, no two such paths share an edge, and every edge appears in at least one such path.

\end{defn}

We note that in the case of a trivial merge tree with a single vertex $v$, the only possible branch decomposition consists of the single degenerate branch $\{ (v,v)\}$.
For a non-trivial merge tree in general position (i.e.~the number of edges adjacent and below any vertex is at most 2), then all saddles and minima (with the possible exception of the root) appear in exactly one pair of the branch decomposition.
If a saddle vertex has $k$ edges below it, then it will occur in $k-1$ pairs.
We can then represent these branches and the relationships between them in a graph of their own as follows.

\begin{defn}[Rooted Tree Representation]
A rooted tree representation $R$ of a branch decomposition $B$ is a graph with vertex set given by $V(R) = B$. 
An edge between $(m,s)$ and $(m',s')$ is included in $E(R)$ if and only if one of the saddles, $s$ or $s'$, is on the path between the other pair. 
That is, either $s$ is on the path from $m'$ to $s'$, or $s'$ is on the path from $m$ to $s$. 
We denote the set of all rooted tree representations of a merge tree $M$ by $S_M$. 
\end{defn}

\noindent See \cref{fig:branch_decompositions} for an example of a merge tree with its possible branch decompositions and the associated rooted tree representations. 

Some graphs may have only a single connected component across all non-empty sublevel sets, so we extend the definition of rooted tree representations to include trivial merge trees. 
A trivial merge tree is  a single vertex with the function value of the lowest vertex of the original graph. 
The trivial merge tree's only rooted tree representation is a single vertex with both elements having the function value of the single vertex of the merge tree.

Branching distance utilizes comparisons between rooted tree representations in order to compare merge trees. 
To compare rooted tree representations $R^X$ and $R^Y$ of merge trees $X$ and $Y$, we form a matching from one vertex set to the other. 
We say that an isomorphism of two rooted trees preserves order when it maps children of a vertex in one tree to the children of its image in the other tree.
Then a matching is an order preserving isomorphism $\gamma:M^X \to M^Y$ for subsets of the vertices $M^X \subseteq V(R^X)$ and $M^Y \subseteq V(R^Y)$. 
The vertices in $M^X$ and $M^Y$ are called matched. 
The remaining vertices are  said to be removed, and are denoted by $E^X = V(R^X) \setminus M^X$ and $E^Y = V(R^Y) \setminus M^Y$. 
A matching is valid when the subgraphs of $R^X$ and $R^Y$ induced by $M^X$ and $M^Y$  respectively are trees and at least one root branch for each has not been removed. 

Having fixed a matching $\gamma:M^X \to M^Y$, we now define the cost of the matching by incurring a cost for each vertex.
The cost of matching two vertices $\gamma(u) = v$, where $ u = (m_u,s_u) \in R^X$ and $v = (m_v,s_v) \in R^Y$, is the maximum of the absolute function value difference of their corresponding elements,
$$\mc(u,v) = \max{(|m_u-m_v|,|s_u-s_v|)}.$$
The cost of removing a vertex $u\in R^X \cup R^Y$ is half the absolute function value difference of the elements of the vertex,
$$\rc(u) = |m_u-s_u|/2.$$

We say that two rooted tree representations $R^X$ and $R^Y$ are $\varepsilon$-$similar$ when we can find a valid matching $\gamma:M^X \to M^Y$ where the maximum cost does not exceed $\epsilon$. 
That is,
\begin{equation*} \label{maxmatch}
    \max_{u\in M^X}\mc(u,\gamma(u)) \leq \varepsilon
\qquad \text{and} \qquad
    \max_{u\in E^X\cup E^Y} \rc(u) \leq \varepsilon.
\end{equation*}
Specifically, the smallest $\varepsilon$ for which the above two inequalities hold is denoted 
$$
\varepsilon_{min}(R^X,R^Y) = 
\min_{\gamma:M^X \to M^Y} \max \left \{
\max_{u\in M^X}\mc(u,\gamma(u)), 
\max_{u\in E^X\cup E^Y} \rc(u)
\right \}
$$

\begin{defn}[Branching Distance]
\label{defn:branchingDistance}

 The  branching distance between two merge trees $X$ and $Y$ is the smallest $\varepsilon_{min}$ out of every possible pair of rooted tree representations, namely
\begin{equation} \label{disteq}
    d_B(X,Y) = \min_{R^X \in S_{X},R^Y \in S_{Y}}(\varepsilon_{min}(R^X,R^Y)).
\end{equation}
\end{defn}
A naive approach to computing this distance would result in an exponential time complexity.
Thus, we use the optimized algorithm described in Beketayev et al.~\cite{morozov}.
The function \isepssimilar determines whether two merge trees can be matched within a given tolerance $\varepsilon$ and is the core of the branching distance algorithm.
The runtime complexity of \isepssimilar is $O(N^2M^2(N+M))$, where $N$ and $M$ are the number of leaves of the input trees.
A binary search with a specified error tolerance is performed to determine $\varepsilon_{min}$, the branching distance.

\begin{figure}
    \centering

    \includegraphics[width=.8\textwidth]{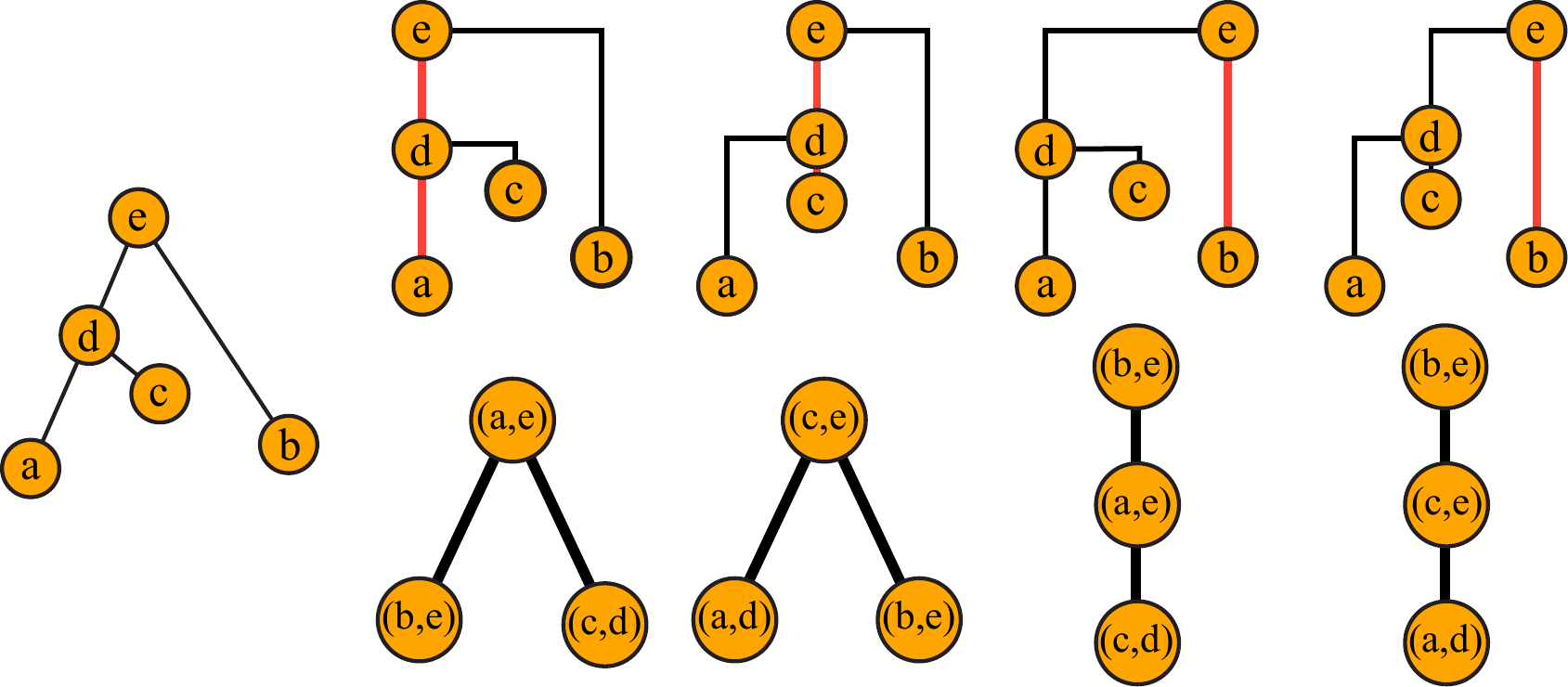}
    \caption{A merge tree $M$ is shown at left. 
    All of $M$'s branch decompositions are shown in top right row, with each root highlighted red. %
    The corresponding rooted tree representations are shown in the bottom right row. 
    }
    \label{fig:branch_decompositions}
    \vspace{-.2in}
\end{figure}

\section{Properties of the Branching Distance}

In this section, we consider and prove properties of the branching distance, showing that it is a semi-metric but not a metric when defined on tail-less merge trees.  
We then define a new distance measure between two geometric graphs that is based on the branching distance, which we call the \emph{average branching distance}.  We prove that the average branching is a distance function, but is not a metric, as two graphs can have average branching distance equal to 0 without being equivalent graphs.

Here, we define distance in the following way, following \cite{setdistancesurvey}. 
\begin{defn}
\label{defn:metric}

Consider a function $d:\cX \times \cX \to \R_{\geq 0}$ on a collection of objects $\cX$. 
Some commonly used properties that such a function might satisfy are: 
\begin{enumerate}[nosep,label=(\roman*)]
    \item Symmetry: for all $x,y \in \cX$,$d(x,y) = d(y,x)$ 
    \item Nonnegativeness: for all $x,y \in \cX$, $d(x,y) \geq 0$ and $d(x,x) = 0$
    \item Positiveness: for all $x,y \in \cX$, $d(x,y) = 0$ implies $x=y$, and 
    \item Triangle inequality: for all $x,y,z \in \cX$, $d(x,y) \leq d(x,z) + d(z,y)$
\end{enumerate}

A function $d:\cX \times \cX \to \R_{\geq 0}$ is called a \emph{distance} if it satisfies (i) and (ii). 
It is a \emph{semi-metric} if it satisfies (i)-(iii), but not necessarily (iv).
It is a \emph{pseudo-metric} if it satisfies (i),(ii), and (iv), but not necessarily (iii).
It is a \emph{metric} if it satisfies (i)-(iv). 

\end{defn}

\subsection{Metric Properties of Branching Distance}
\label{BD:semi-metric}

In this section, we will show that the branching distance on tail-less merge trees satisfies properties (i)-(iii) above, making it a semi-metric. 
However, we will also provide a counter example to the triangle inequality, proving that the branching distance is not a metric. 
Throughout this section, let $X$, $Y$, and $Z$ be merge trees. 
We will abuse notation by writing $m$ and $s$ for both the vertices in the merge tree that they represent, and for the function value of those vertices.
We first show symmetry. 
To show that $d_B(X,Y) = d_B(Y,X)$, we will show that we obtain the same minimum $\varepsilon_{min}$ when comparing $X$ to $Y$ as we obtain when comparing $Y$ to $X$.

\begin{lemma}
\label{lem:sameMineps}
$\varepsilon_{min}(R^X,R^Y) = \varepsilon_{min}(R^Y,R^X)$
\end{lemma}

\begin{proof}

Let $u = (m_u,s_u) \in R^X$ and $v = (m_v,s_v) \in R^Y$ be vertices. 
Because of the symmetry of the absolute value expressions, we have
\begin{equation*}
\begin{split}
\mc(u,v) & = \max(|m_u-m_v|,|s_u-s_v|)\\
& = \max(|m_v-m_u|,|s_v-s_u|)\\
& = \mc(v,u).
\end{split}        
\end{equation*}

This means that 
$\max_{u \in M^X}\mc(u,\gamma(u)) = \max_{v \in M^Y}\mc(v,\gamma(v))$, 
so we obtain the same $\varepsilon_{min}$ for the matching costs. 
By definition of union, 
$\max_{u \in E^X \cup E^Y}\rc(u) = \max_{u \in E^Y \cup E^X}\rc(u)$, 
so we obtain the same minimum $\varepsilon$ for the removal costs. 
The $\varepsilon_{min}$ are the same regardless of order, so the minimum $\varepsilon$ that satisfies both will be the same regardless of the order we compare the rooted branchings.
\end{proof}

\begin{corollary}
\label{banching_symmetric}
$d_B(X,Y) = d_B(Y,X)$
\end{corollary}

\begin{proof}
Because our function considers all possible edits, and by \cref{lem:sameMineps}, the minimum satisfying $\varepsilon$ is the same regardless of the order of the rooted branchings, the cheapest edit from $X$ to $Y$ will be the cheapest edit from $Y$ to $X$, and the minimum $\varepsilon_{min}$ will be the same.
\end{proof}

We next show nonnegativeness of the branching distance. 

\begin{lemma}
\label{branching_positive}
$d_B(X,Y) \geq 0$
\end{lemma}
\begin{proof}
 The distance between any two graphs is non-negative as all costs are non-negative due to the absolute value expressions.
 \end{proof}

\begin{lemma}
\label{branching_samezero}
$d_B(X,X) = 0$
\end{lemma}
\begin{proof}
 To show that $d_B(X,X) = 0$, we will show that if $X = Y$ then $d_B(X,Y) = 0$.
Suppose $X = Y$, then there is an isomorphism $\gamma: V(X) \to V(Y)$ which preserves function values. 
We select branches $(u,v) \in X$ and $(\gamma (u), \gamma (v)) \in Y$ for our branch decompositions such that all relations are preserved and corresponding vertices have equivalent function values. 
This means $\max{ mc(u,\gamma(u))} = 0$  and we do not need to make any vertex deletions, so  $\max{rc(u)} \leq 0$. 
Hence, $d_B(X,Y) = 0$.
\end{proof}

\begin{lemma}
\label{dist_lemma:1}
If $d_B(X,Y) = 0$ then $X = Y$
\end{lemma}
\begin{proof}
Suppose $d_B(X,Y) = 0$. 
Then there exists a pair of rooted branchings $R^X \in B_{X}$ and $R^Y \in B_{Y}$, such that
\begin{equation*}
\max_{u \in M^X}\mc(u,\gamma(u)) \leq 0
\qquad \text{and}\qquad 
\max_{u \in E^X \cup E^Y}\rc(u) \leq 0.
\end{equation*}
If $\max_{u \in M^X}mc(u,\gamma(u)) = 0$ then, for every pair of matching vertices $u$ and $v$, 
$$\max{(|m_u -m_v|,|s_u-s_v|)} = 0.$$ 
This means that all corresponding vertices in the merge trees must have equivalent function values.

The requirement that $\max{\rc(u)} = 0$ is only satisfied when either all removals cost $0$, or when we do not need to remove any vertices. 
A removal cost of $0$ would mean that a minimum has the same function value as its saddle for some non-root branch, which would not be included in a merge tree, hence the rooted tree representations must contain the same number of vertices.
    
We therefore have a bijective mapping between the vertices of $X$ and $Y$ that preserves function values and the relations between nodes. 
We consider isomorphic graphs with identical function values to be equivalent, hence $X = Y$.
 \end{proof}
 
 \begin{figure}
    \centering
    \includegraphics[height = 2in]{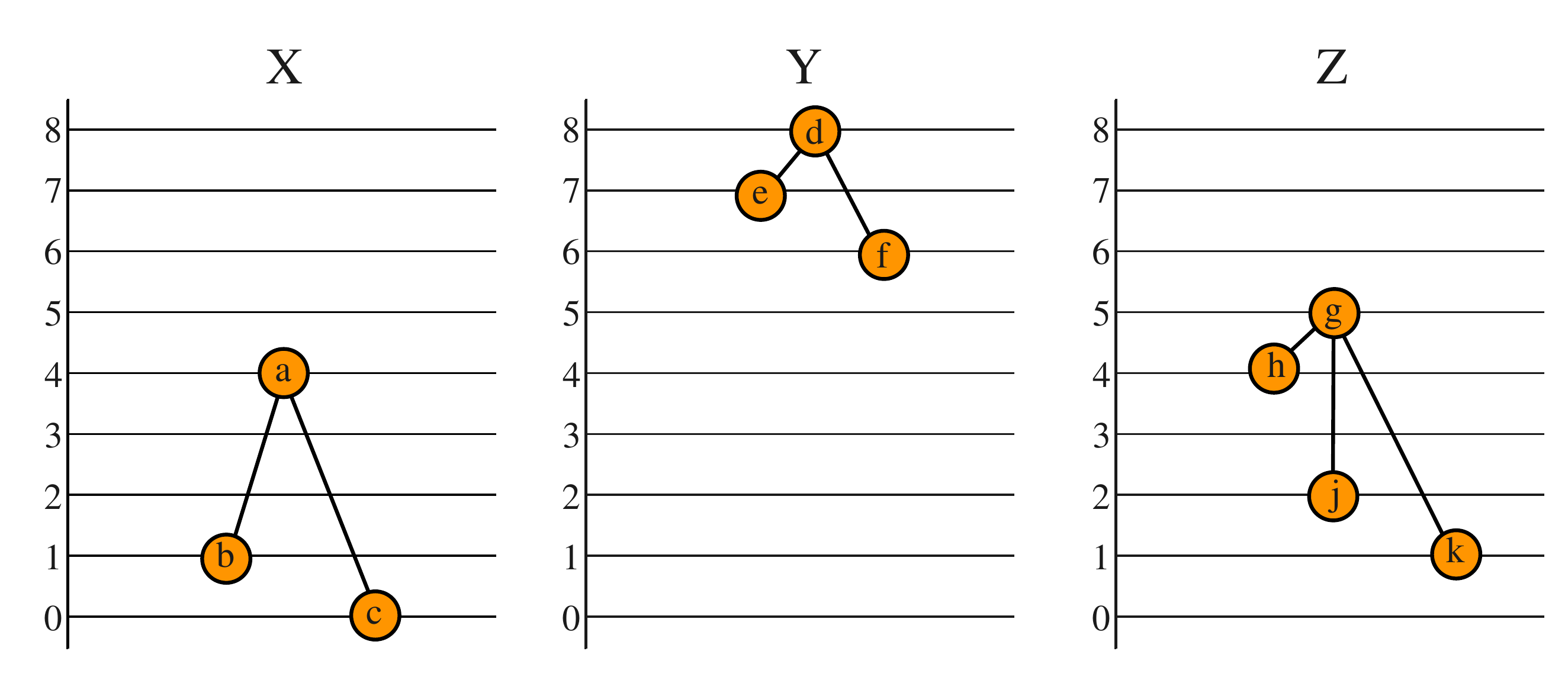}
    \caption{Merge Trees for the counterexample to the triangle inequality}
    \label{fig:merge_counter_example}
\end{figure}

We next give a counterexample to the triangle inequality.
 
\begin{theorem}
\label{thm:BranchingNotTriangle}
The branching distance does not satisfy the triangle inequality, and therefore is not a metric. 
\end{theorem}

\begin{proof}

To prove this, we provide a counterexample to the triangle inequality; that is, a collection of three merge trees for which $d_B(X,Y) \geq d_B(X,Z) + d_B(Z,Y)$. 

Consider the merge trees in  \cref{fig:merge_counter_example}.
When comparing $X$ and $Y$ we observe that the removal cost of any of the nodes will be less than the matching cost, so to determine the optimal edit from $X$ to $Y$ we need to select the root branches with the lowest comparison cost. 
We see that leaves $b$ and $f$ are the closest in function value, so our optimal edit will select $(b,a)$ as our root branch for $X$, $(f,d)$ as our root branch for $Y$, and delete $(c,a)$ and $(e,d)$. Our most expensive edit results from comparing the root branches and gives a minimum $\varepsilon$ of $5$, hence $d_B(X,Y) = 5$.

When comparing $Y$ and $Z$, we observe that the removal cost of the nodes will also be less than any matching cost, so to determine the optimal edit from $Y$ into $Z$ we need to select the root branches with the lowest comparison cost. 
We see that leaves $f$ and $h$ are closest in function value, so our optimal edit will select $(f,d)$ as our root branch for $Y$, $(h,g)$ as our root branch for $Z$, and delete $(e,d)$, $(j,g)$, and $(k,g)$. 
Our most expensive edit results from comparing the root branches and gives a minimum $\varepsilon$ of $3$, hence $d_B(Y,Z) = 3$.

When comparing $X$ and $Z$, we observe that removing $(g,h)$ will always be cheaper than comparing it, so we have to select a different root branch for $Z$. 
We also observe that the removal cost of any remaining vertex will be more expensive than comparing it, so we can form an optimal edit by removing $(h,g)$ and comparing the remaining vertices with the closest function values. 
We select either root branch for $X$ and compare $(b,a)$ to $(j,g)$ and $(c,a)$ to $(k,g)$ to obtain an optimal edit with a minimum $\varepsilon$ of $1$, hence $d_B(X,Z) = 1$.

Putting this together, we have that $d_B(X,Y) = 5$, $d_B(Y,Z) = 3$, $d_B(X,Z) = 1$, and hence 
\[ 
d_B(X,Y) = 5\geq 3 + 1 =  d_B(X,Z) + d_B(Z,Y)
\]
contradicting the triangle inequality.
As a consequence of this, the branching distance is not a metric.
\end{proof}

\subsection{Average Branching Distance}
\label{ABD}

Branching distance (Defn.~\ref{defn:branchingDistance}) is defined for comparing a pair of  tail-less merge trees, but our goal is to compare graphs with a function to  $\mathbb{R}^2$. 
We could take the merge tree of the graph for a single direction and use the branching distance, but this would result in a large amount of spatial information being be neglected, and there is no obvious \say{best} direction. 
Our solution is to compute the branching distance at some satisfactorily large amount of equally spaced angles and report the median of these measurements.

Since our rotation is about the origin, the range of function values is distorted in a way that disrupts branching distance's ability to compare merge trees. 
However, by shifting the function values when we calculate the merge tree, we are able to ensure that they stay centered about zero.
In order to use branching distance, we first construct a merge tree corresponding to a subset of directions $\omega \in [0,\pi)$ as described in Sec.~\ref{ssec:DirectionTransformMergeTree}.
Note that as described in that section, we shift the merge tree function values to have an median value of 0 in order to compare the trees based on overall structure even when the functions are vastly different. 

In our implementation, we begin with starting orientation $\omega = \pi/2$ and compute the merge tree as given by Defn.~\ref{defn:MergeTree}.
Then, to compare two graphs, we rotate both graphs simultaneously, and compute the merge tree for each in the given direction. 
The result is the average branching distance as follows.

\begin{defn}[Average branching distance (ABD)]
\label{defn:ABD}
Fix directions 
$\Omega = \{\omega_1,\cdots,\omega_n\} \subseteq [0,2\pi)$
and let $\mathrm{avg}(X)$ be fixed as either the mean or median of the set $X \subset \R$. 
Let $M_{G,\omega_i}$ be the merge tree of $G$ rotated by angle $\omega_i$. 
The average branching distance between two graphs $G$ and $H$, $d_A(G,H)$,
is defined to be 
\begin{equation*}
    d_A(G,H) = \mathrm{avg}(\{ d_B( M_{G,\omega_i},M_{H,\omega_i} \mid i = 1,\cdots, n\} )
\end{equation*}
That is, $d_A(G,H)$ is either the median or the mean of the set of  branching distances of their merge trees computed over a specified set of angles.
\end{defn}

\subsection{Metric Properties of Average Branching Distance}

In this section, we discuss properties of the average branching distance. 
In particular, we show that it is (i) symmetric, and (ii) nonnegative as given in Defn.~\ref{defn:metric}. 
We provide a counterexample to (iii) positiveness, as well as to (iv) the triangle inequality, which shows that it is not a metric, semi-metric, or pseudometric. 
For the rest of the section, let $G$, $H$, and $J$ be arbitrary embedded graphs. 

\begin{lemma}
$d_A(G,H) = d_A(H,G)$
\end{lemma}

\begin{proof}
We use the same set of angles to measure the distance and at each angle we are comparing the same pair of merge trees. This means we are always comparing the same pairs of merge trees for $d_A(G,H)$ and $d_A(H,G)$.
By Cor.~\ref{banching_symmetric} we know $d_B(M_G,M_H) = d_B(M_H,M_G)$ for any merge trees $M_G$ and $M_H$. 
This means the ordered set of distances for $d_A(G,H)$ is equivalent to the ordered set of distances for $d_A(H,G)$, so it follows that the medians of the sets will be equivalent.
\end{proof}

\begin{lemma}
$d_A(G,H) \geq 0$
\end{lemma}

\begin{proof}
By Lem.~\ref{branching_positive} we know that $d_B(M_G,M_H) \geq 0$ 
for any merge trees $M_G$ and $M_H$. 
This means every element of the ordered set of distances for two graphs is non-negative, hence the median of the set must be non-negative. 
\end{proof}

\begin{lemma}
$d_A(G,G) = 0$
\end{lemma}

\begin{proof}
We are comparing the same orientations of identical graphs at each angle, so the merge trees being compared will be identical for any orientation. 
By Lem.~\ref{branching_samezero}, $d_B(M_G,M_G) = 0$ for any merge tree $M_G$, hence every element of any ordered set of distances for $d_A(G,G)$ is $0$ and the median of the set will be $0$.
\end{proof}

\begin{figure}
    \centering
    \includegraphics[scale = 0.3]{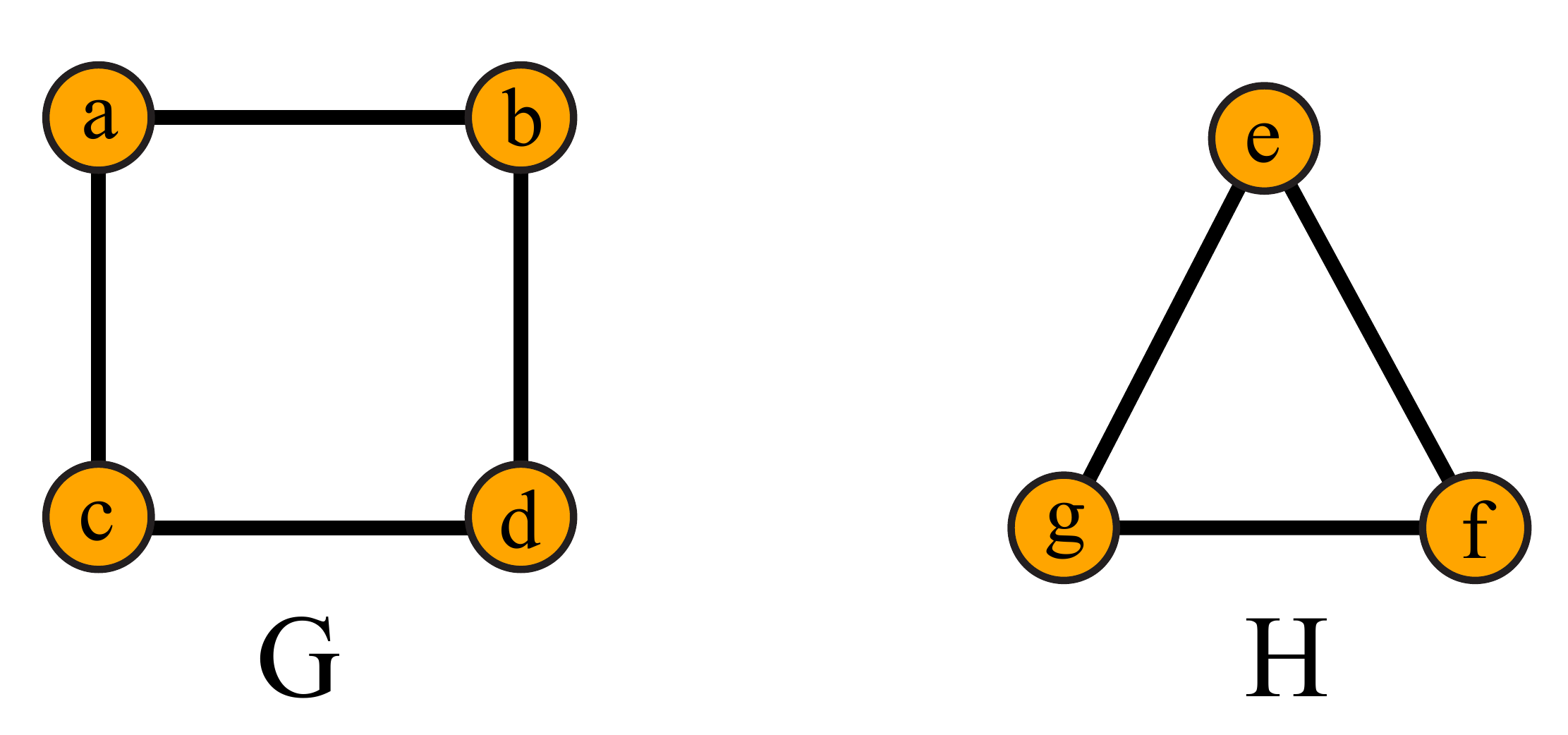}
    \caption{Graphs for the counterexample to $d_A(X,Y) = 0$ implies $X = Y$.}
    \label{fig:centered_counter_example}
\end{figure}

We will now provide a counter-example to (iii) positiveness from Defn.~\ref{defn:metric}. 
\begin{theorem}
\label{thm:ABD_Fails_iii}
There is a pair of non-isomorphic graphs $G$ and $H$ for which $d_A(G,H) = 0$.
\end{theorem}

\begin{proof}

Consider the non-isomorphic graphs $G$ and $H$ in Figure \ref{fig:centered_counter_example}.
We see that at any angle of comparison we are comparing two trivial merge trees which will always be shifted to function value $0$ after their construction. The branching distance between any such pair of merge trees will be $0$ hence, every element of the ordered set of distances for $d_A(G,H)$ will be zero and the median of all these distances is $0$. 
However, $G$ is not equivalent to $H$ as the graphs are non-isomorphic. 
As a consequence of this, average branching distance is neither a metric nor a semi-metric. 
\end{proof}

We will now provide a counter example to the triangle inequality, $d_A(G,H) \leq d_A(G,J) + d_A(J,H)$. 

\begin{theorem}
\label{thm:ABD_fails_triangle}
There is a collection of graphs $G$, $H$, and $J$ for which $d_A(G,H) > d_A(G,J) + d_A(J,H)$. 
\end{theorem}

\begin{proof}
Let the set of angles be $\{\frac{\pi}{2}\}$, and let $G$,$H$, and $J$ be the graphs in Figure \ref{fig:abd_counter}. 
The graphs' shifted merge trees $M_G,M_H$ and $M_J$ are identical to the original graphs at the standard orientation and we are only comparing at one angle, hence each pairwise ABD over $\{\frac{\pi}{2}\}$ will be the same as the respective branching distance.

\begin{figure}
    \centering
    \includegraphics[scale=0.4]{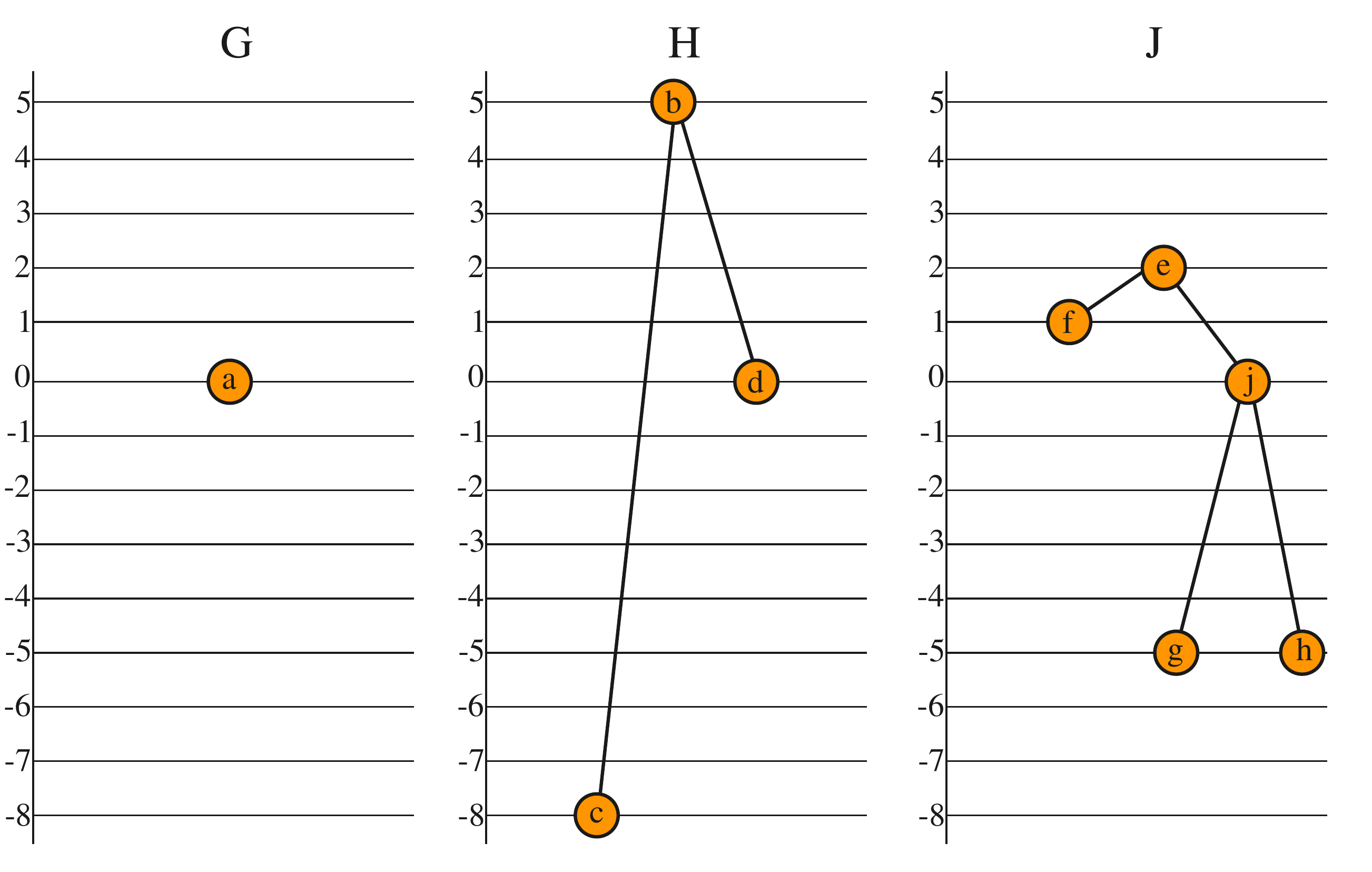}
    \caption{Graphs for the counterexample to the triangle inequality for ABD.}
    \label{fig:abd_counter}
\end{figure}
When comparing $M_G$ and $M_H$ we observe that $M_G$ is a trivial merge tree, hence the only rooted tree representation for $M_G$ is a single root branch $(a,a)$. Matching branch $(c,b)$ is more expensive than removing it, so we remove $(c,b)$ and match the single root branch of $M_G$ to $(d,b)$. The maximum cost for this optimal edit is $6.5$ from removing $(c,b)$, hence $d_A(G,H) = d_B(M_G,M_H) = 6.5$.

When comparing $M_G$ to $M_J$, we observe that matching any branch containing $g$ or $h$ would be more expensive than removing it, hence we must choose either $(f,e)$ or $(j,e)$ as our root branch and remove all other branches. For either selection of root branch, the most expensive cost is $2.5$ from removing $(g,j)$ and $(h,j)$, hence $d_A(G,J)= d_B(M_G,M_J)=2.5$.

When comparing $M_H$ to $M_J$, we observe that matching $(c,b)$ to a branch containing $g$ or $h$ will be cheaper than removing it. 
We also see that we need to select $(h,e)$ or $(g,e)$ as our root branch in order to match $(c,b)$ to a branch containing $g$ or $h$. 
Removing $(g,j)$ or $(h,j)$ will always be cheaper than matching it to $(b,d)$, so we remove the one we didn't select as our root branch. 
Matching or removing $(f,e)$ and $(d,b)$ does not effect our maximum cost, hence any optimal edit will have a maximum cost of $3$ from matching the root branches. Therefore, $d_A(H,J)=d_M(M_H,M_J)=3$.

We see that $d_A(G,H) = 6.5$, $d_A(G,J) = 2.5$, $d_A(H,J) = 3$, and $6.5 > 2.5 + 3$ contradicting $d_A(G,H) \leq d_A(G,J) + d_A(H,J)$. 

\end{proof}

Combining \cref{thm:ABD_Fails_iii} and \cref{thm:ABD_fails_triangle} gives us that ABD is neither a semi-metric nor psuedo-metric, thus it is only a distance.

\subsection{Average Branching Distance of Convex Polygons}
In the process of simplifying our data through merge tree construction, we lose information such as interior complexity and some shape data.  
We will show that average branching distance will consider any convex polygons to be similar, and give an example of data  not represented well by our function.

\begin{lemma}
The merge tree from any direction of a convex polygon graph is trivial.
\label{lemma:complex_trivial}
\end{lemma}
\begin{proof}
For any orientation of a graph $G$, there will be some minimum function value $k$.
Any two vertices with function value $k$ must be connected by a horizontal path, otherwise we would need some higher function value reflex vertex connecting the two vertices, contradicting the fact that G is a convex polygon.
Every other vertex must be degree two with a descending path to some bottom vertex with function value $k$, hence the entire graph is one connected component regardless of sublevel set.
\end{proof}

\begin{corollary}
$d_A(G,H) = 0$ when $G$ and $H$ are graphs of convex polygons.
\end{corollary} 

\begin{proof}
By Lemma \ref{lemma:complex_trivial} we know that at any angle the entire merge tree for either graph is a trivial merge tree, hence we will always be comparing two trivial merge trees which will always be shifted to function value $0$ after their construction.
The branching distance between any such pair of merge trees will be $0$, hence every element of the ordered set of distances for $d_A(G,H)$ is $0$, and the median 
of this set is $0$.
\end{proof}

This can be extended to graphs with convex polygon borders with no merges inside the polygon. 
For example, consider the graphs in Fig.~\ref{fig:convex_polygons} from the IAM database \cite{iam}:
From a centered merge tree perspective, these graphs are identical from any angle, but we can clearly see that there are significant differences between them due to their shape and interior data.
\begin{figure}
    \centering
    \includegraphics[scale = 0.4]{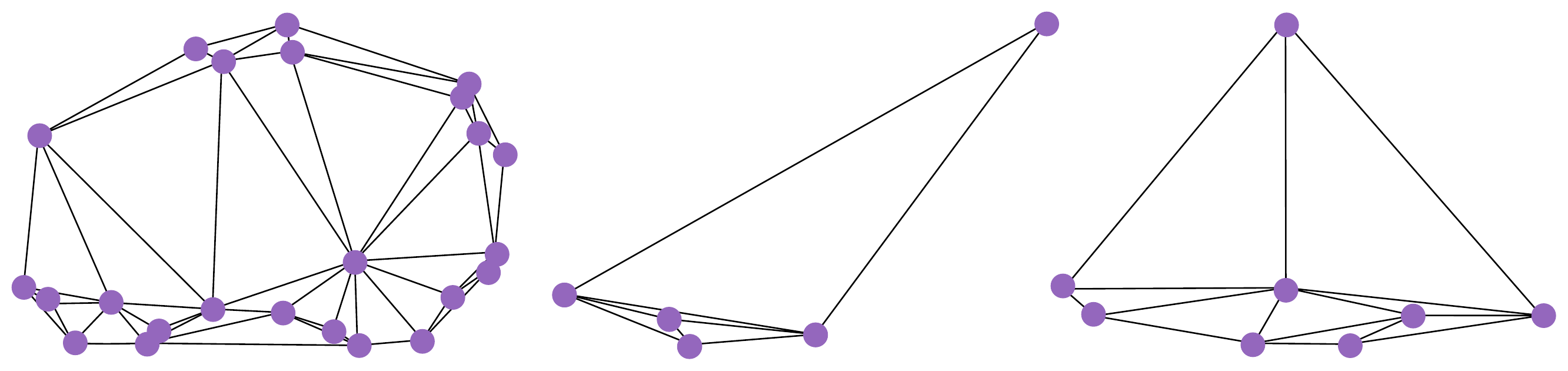}
    \caption{Graphs from the IAM database \cite{iam} which have convex polygon borders}
    \label{fig:convex_polygons}
\end{figure}

\section{Implementation, Experimentation, and Results}

To explore the capabilities of our proposed distance function, we developed Python code (version 3.7.6) to compute the average branching distance. 
We implemented the algorithm presented in Beketayev et al.~\cite{morozov} to compute the branching distance.
In order to use merge tree distances as a method for graph comparison, we also implemented an original method for merge tree construction, which we describe in \cref{sec:mergetreealg}.
For tasks that required interactions with graphs, we used  the NetworkX Python package~ \cite{networkx}; 
we also utilized several other scientific computing packages \cite{matplotlib,scikit-learn,numpy,scipy}. 
Our code  can be found in version 1.0 of our GitHub repository~\cite{us}.

\subsection{Merge Tree Algorithm}
\label{sec:mergetreealg}

The design of our merge tree construction algorithm was motivated by the union-find data structure, which is commonly used to track elements of a set which are divided into disjoint subsets.
Specifically, our algorithm keeps track of child and parent pointers to identify connected components at the graph's sublevel sets. 
Nodes in the original graph will be assigned child pointers, whereas nodes in the merge tree will be assigned parent pointers.
We note that we use the term pointer here in the computer science sense, where, for each vertex, we are essentially associating a unique other vertex in the graph based on some desired property.

Fix a subgraph $H \subseteq G$, and note that this graph potentially has more than one connected component. 
We define the \textit{representative} $v_r$ of a vertex $v$ to be the lowest function valued vertex connected to $v$ in the subgraph $H$.
Assuming a generic function, i.e.~all vertices have distinct function values, all nodes in the same component share the same, unique representative\footnote{Note that our code does not require genericity, but the assumption simplifies the discussion of the algorithm.}.
For this reason, we also call $v_r$ the representative of the component containing $v$.

To find this representative, for each vertex $v \in V(H)$, we associate a \textit{child pointer} which is in the same connected component as $v$ but not necessarily adjacent. 
The representative is the only node in a component which is its own child pointer.
The child pointers are set so that the representative of a vertex can be found by following child pointers until a node which is its own child pointer is found.

The \textit{root} $v_p$ of a vertex $v$ is the node with the largest function value to which it is connected.
Like the component representative, the root is the same for all vertices in a connected component. 
The \textit{lower neighbors} of $v$ are all vertices $u$ adjacent to $v$ such that $f(v) > f(u)$; \textit{upper neighbors} are defined symmetrically.
The \textit{parent pointer} of a vertex in a merge tree is itself if it is the root; or its only upper neighbor otherwise.
Similar to the representative, the root can be found by following the parent pointers of $v$ until a vertex with is its own parent pointer is found.

We will use a sweep line algorithm, where we process the graph vertices in order of function value $a \in \R$.
Setting $H_a = f\inv(-\infty,a] $ to be the sublevelset of the graph, we maintain the merge tree of $H_a$ at each step as we increase $a$.
As the merge tree is constructed, the child pointers of relevant nodes in $H_a$ will be updated; and parent pointers will be maintained in the merge tree being constructed.
The later constitute the edge set of the newly constructed merge tree.
This will be done to ensure that all parent-paths and child-paths in a connected component lead to the correct root and representative, respectively. 

We assume that our input graph has been pre-processed as follows. 
During this step, all adjacent nodes with the same function value will be collapsed into a single node, thus we can assume that our graph input has unique function value for adjacent vertices. 
This will not alter the merge tree and ensures that all edges have a specified upper and lower vertex.

\begin{breakablealgorithm}
  \caption{Create merge tree}
	\begin{algorithmic}[1]
	\Require{A graph $G=(V,E)$ with function $f:V \to \R$ such that $f(u) \neq f(v)$ for all $uv \in E$}
	    \State Initialize an empty graph $M = (V(M),E(M))$ with $V(M) = \emptyset$ and a function $f_M: V(M) \rightarrow \mathbb{R}$
	    \State Note that each vertex in $V(M)$ will be associated to a unique vertex in $V(G)$ which caused its addition.
	    \State Initialize all vertices in $G$ with an empty child pointer, and an empty parent pointer.
	    \State Sort $V$ ascending by function value
		\For {Vertex $v \in V$ with $f(v) = a$}
		    \State Let  $R \subset V$ be the representatives of the components of  $v$'s lower neighbors in $H_a = f^{-1}(-\infty,a] \subseteq G$, found by recursively following children pointers until a vertex is found which is its own child.
		    \State Let $R_M \subset V(M)$ be the vertices in $M$ corresponding to the elements of $R$ 
            \State
    		\If{$|R|=0$}
        		\State Add a new vertex $v_M$ to $V(M)$ 
            	\State Set $f_M(v_M)=f(v)$
        		\State Set $v_M$ as its own parent
        		\State Set $v$ as its own child
    		\EndIf
    		\State
    		\If{$|R|=1$}
    		    \State Set the one element in $R$ as $v$'s child
    		\EndIf
    		\State
    		\If{$|R|\geq2$}
                \State Choose the vertex $c \in R$ with minimum function value, i.e.~$ f(c) \leq f(u)$ $ \forall u \in R$
        		\State Make a copy $v_M$ of $v$ in the merge tree $M$
        		\State Set $f_M(v_M)=f(v)$
        		\State For all $x \in R_M$, find the root of the connected component in $M$, $x_p$ 
        		\State  Add the edge $x_p v_M$  to $E(M)$ for each $x_p$ and set the parent of $x_p$ to be $v_M$.
        		\State Collapse all on-level neighbors of $v_M$
        		\State Set $v_M$ as its own parent
        	    \State Set $c$ as the child of $v$ and all $x \in R$
            \EndIf
        \EndFor
        
        \State Return $M$ and $f_M$
	\end{algorithmic}
\end{breakablealgorithm}

Notice that by the definition of a merge tree, $\Delta(a)$ will only be non-empty when a minimum is encountered or components merge. 
Lines 9-14 guarantee that all minima are detected and included in the merge tree. 
Furthermore, lines 20-29 guarantee that all component merges are detected and included in the merge tree. 
Vertices in a merge tree are only connected when components merge. 
Line 25 accounts for all the connections between vertices. 
Trivially, $f_M$ is correctly defined. 
Therefore, the proposed algorithm correctly constructs a merge tree based on Definition \ref{M}.

\subsection{Experiments}\label{experimentation}

\begin{figure}
    \centering
    \begin{minipage}{.4\textwidth}
    \centering
    \includegraphics[align=c, width= .8\textwidth]{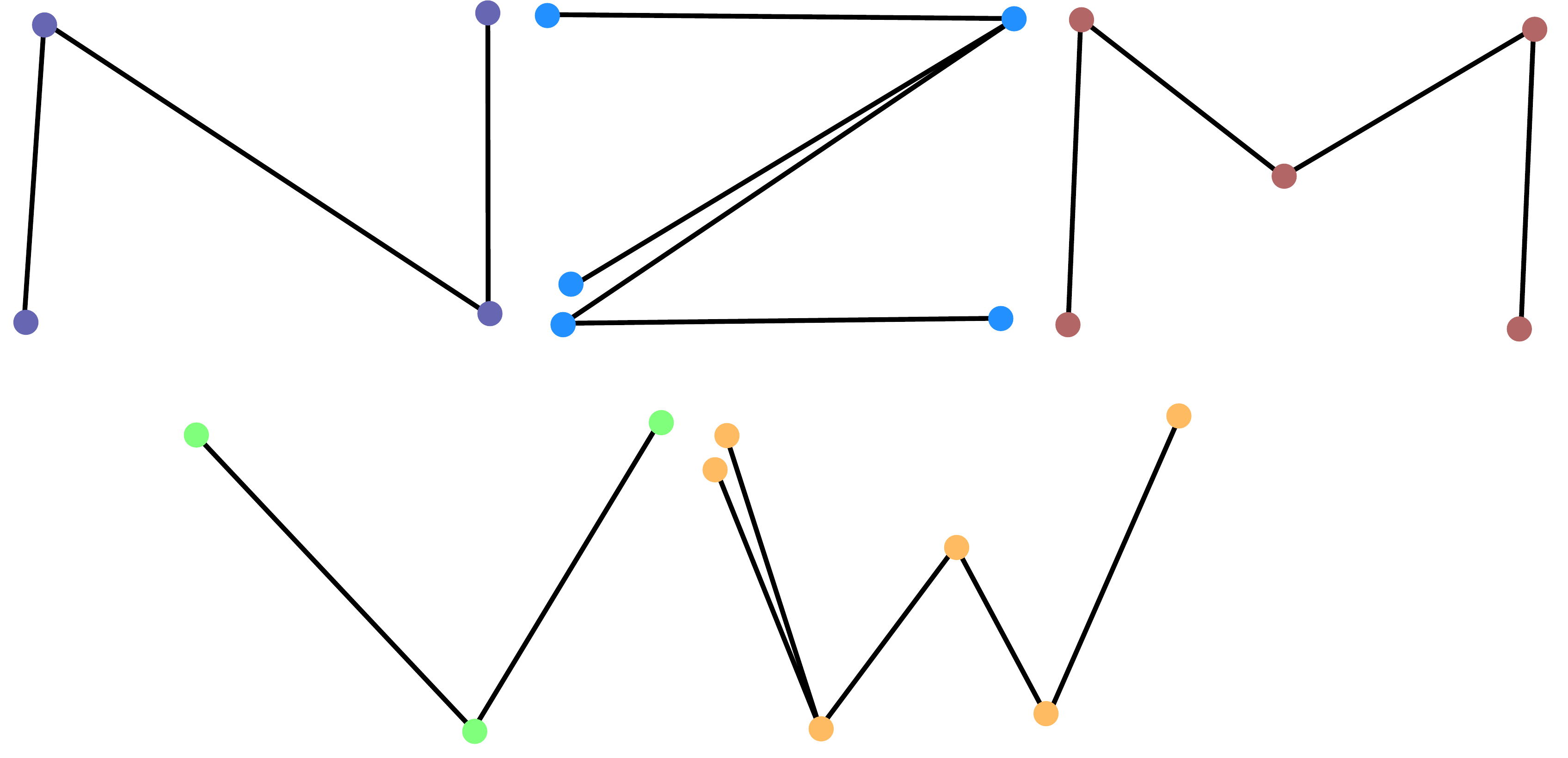} 
    \includegraphics[align=c, width= \textwidth]{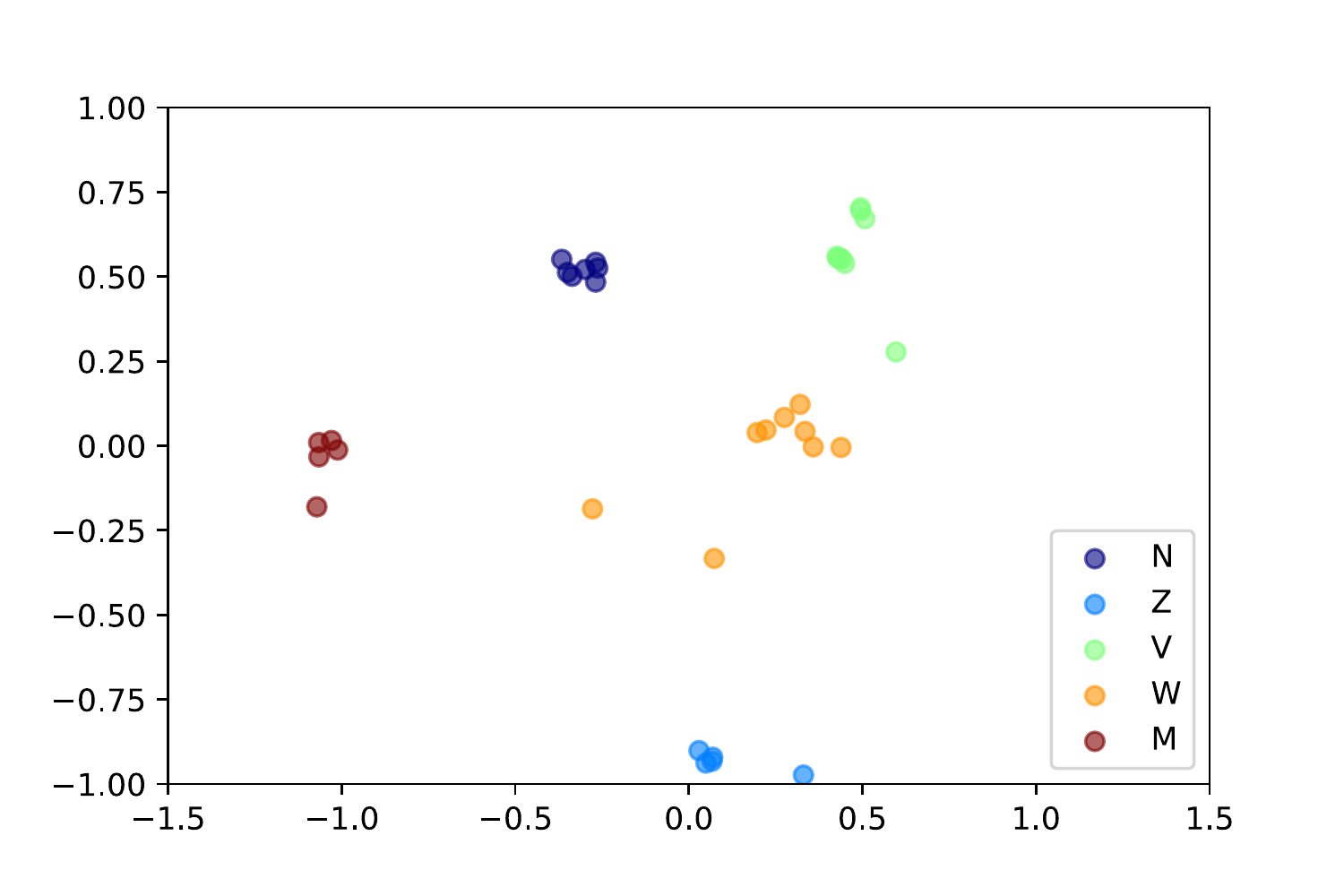}
    \end{minipage}
    \includegraphics[align=c, width = .59\textwidth]{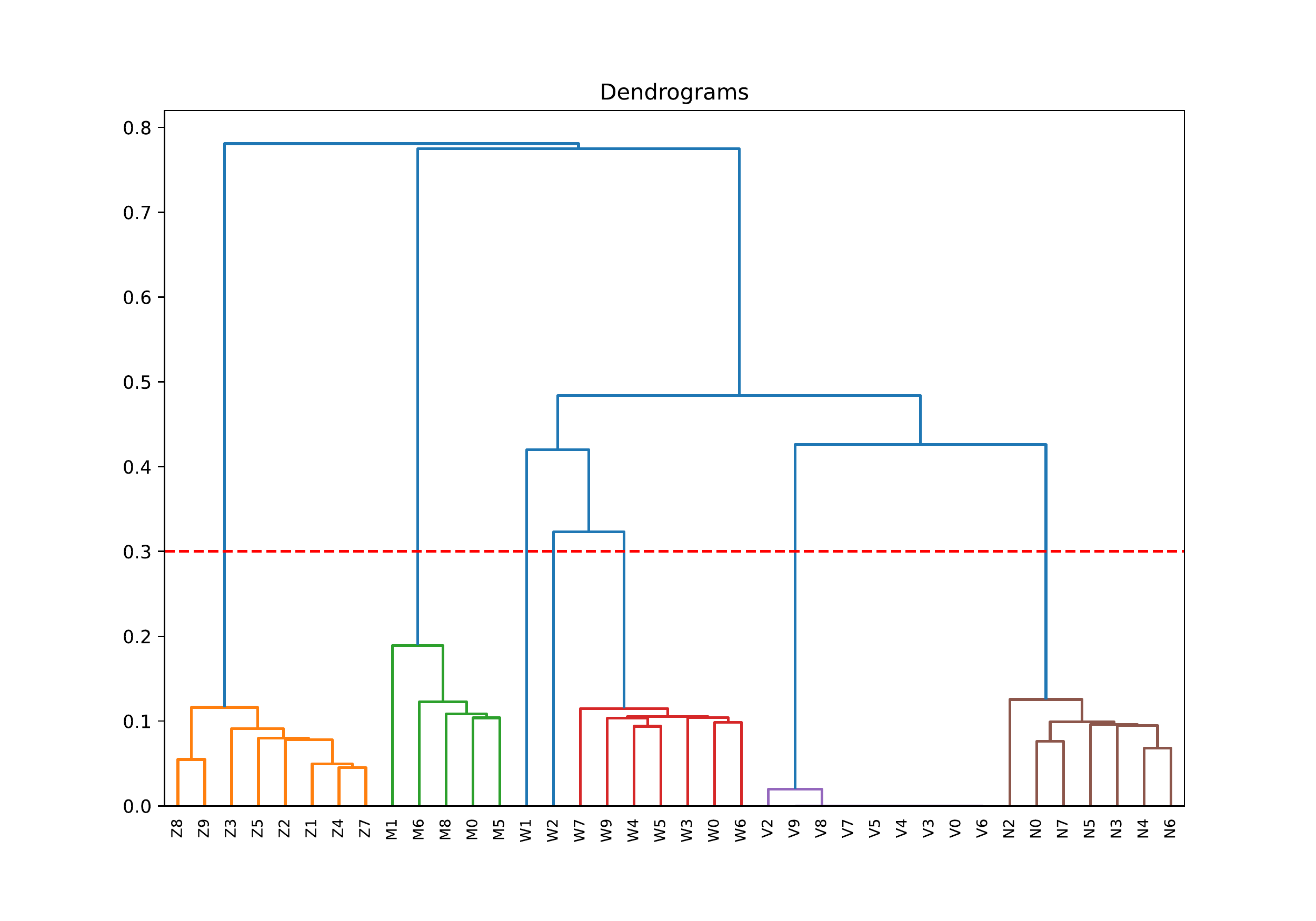}
    \caption{Clockwise from top left: A sample of letters from the IAM database; results of single-linkage hierarchical clustering on 38 letter graphs; MDS plot of the same graphs.}
    \label{fig:dendro}
\end{figure}
In this section, we describe our experiments to determine if ABD is a good measure of dissimilarity between graphs.
We are looking to see whether graphs that appear similar are considered \say{close} to each other. 
To do this, we compute the ABD between several test data sets and use cluster analysis and dimensionality reduction to visualize whether our distance function upholds this idea.
The plots shown in this section can be reproduced by running \texttt{Plots.py} from version 1.0 of our GitHub repository \cite{us}.

We specify the orientations at which the branching distance between two graphs is calculated for ABD with the term \textit{frames}. 
The \textit{n} frames used to compute ABD is determined by the \textit{n} evenly spaced orientations covering the interval $[0,2\pi)$, where \textit{n} is a positive integer parameter of the ABD function. 
For the following examples we use a small number of frames due to the simplicity of the input graphs. 
We do so because we do not anticipate significant changes between merge trees of the same graph if these graphs are rotated only slightly. 
For more intricate input graphs, such as maps, we would suggest a higher number of frames. 

We consider graphs of different letters from the IAM Graph Database \cite{tud, iam}.  
The graphs in this data set represent distorted letter drawings, so it was necessary for us to exclude graphs that we considered distorted beyond the point of recognition. 
Specifics of the process for identifying these outliers are documented in \texttt{data/DataCleaning.py} in version 1.0 of our repository.
We use ABD on two sets of these graphs to compute pairwise distances. 
Because the letter graphs are so simple, we only compute the branching distance at 10 frames.
We pass these values through a single linkage hierarchical clustering algorithm and construct a dendrogram to examine the results, as shown in Figure \ref{fig:dendro}. 
This provides a visual representation of which graphs are most similar, which is more clear than looking at a matrix of values.
As seen in the right of Figure \ref{fig:dendro}, distinct graphs of the same letter have closer links between them than graphs of different letters. 
There are clear groupings of letters visible in the dendrogram, signified by color. 
Graphs that appear similar to human eyes to have smaller distances between them, which provides support for the validity of average branching distance as a measure of dissimilarity.  
We can see these result in another format if we pass the same distance matrix through a multidimensional scaling algorithm \cite{Cox2000}, the results of which are shown in the bottom left of Figure \ref{fig:dendro}. 
We see that letters whose appearances we can verify as similar are closer to each other on the scatter plot. 

\begin{figure}
    \centering
    \subfloat[Binary images]{{\includegraphics[width=5cm]{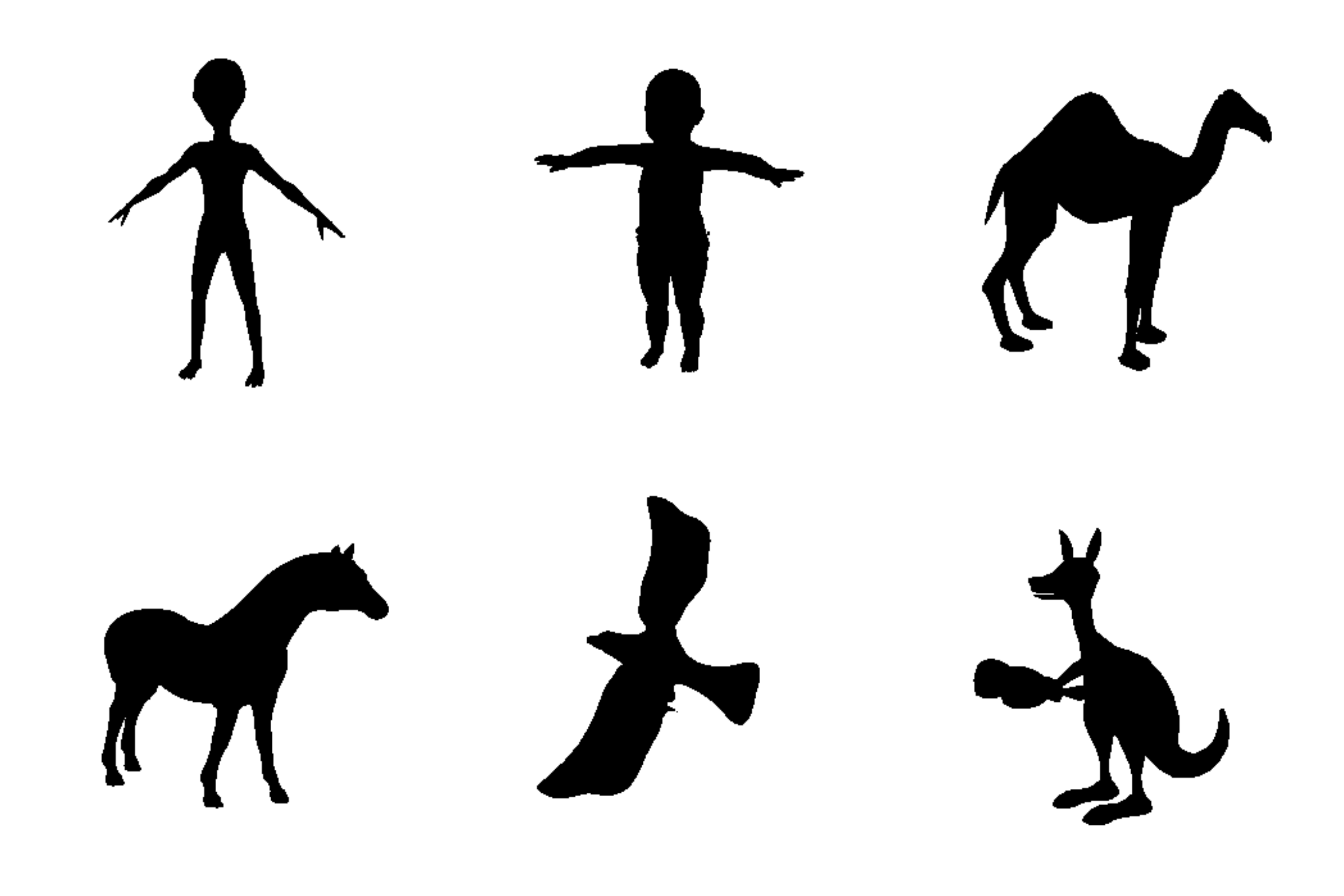} }}%
    \qquad
    \subfloat[Graphs of skeletonized images]{{\includegraphics[width=7cm]{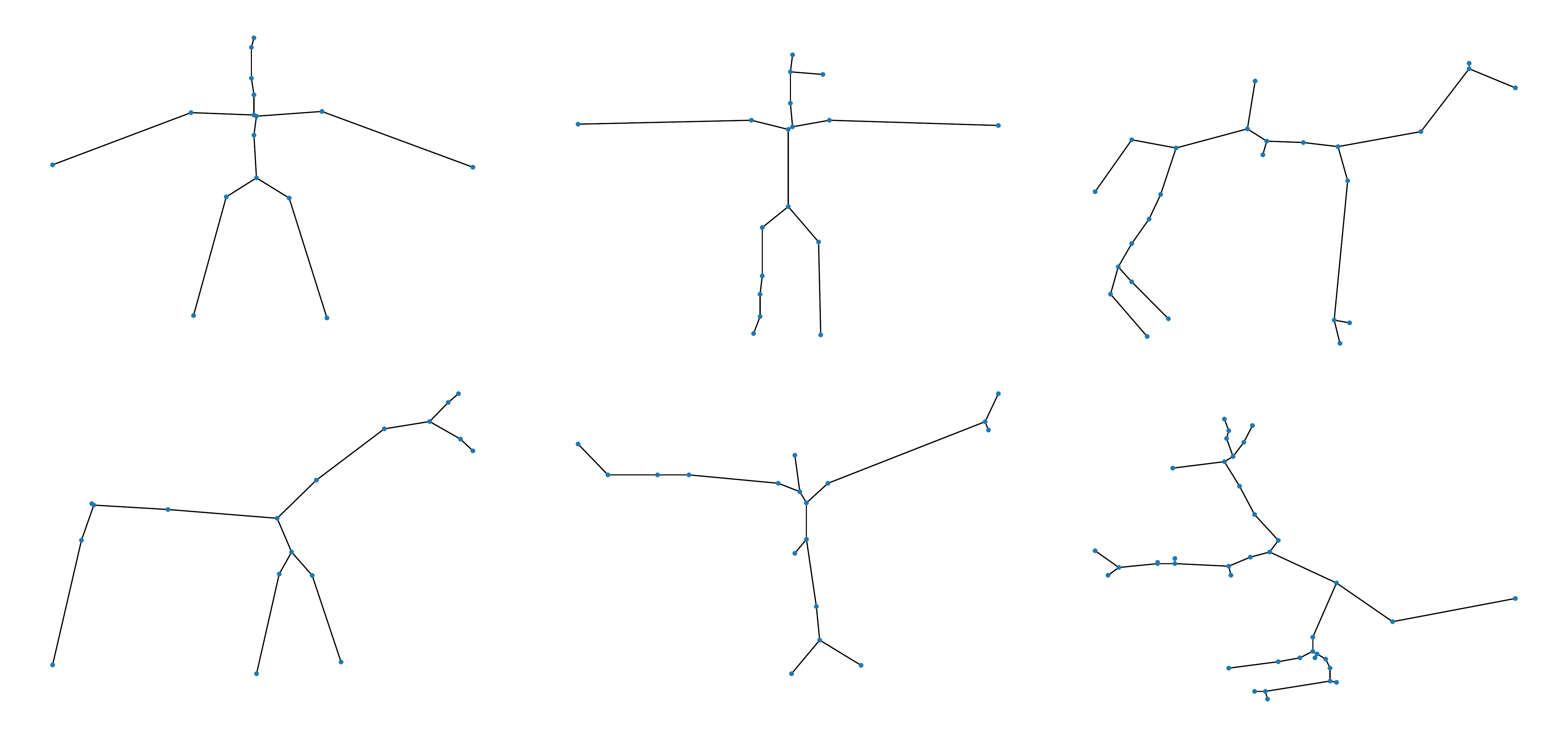} }}%
    \caption{Clockwise from top left; alien, child, camel, kangaroo, eagle, horse (not to scale) \cite{shapematcher}}%
    \label{fig:binaryandskel}%
\end{figure}

\begin{figure}
    \centering
    \subfloat[5 Frames]{{\includegraphics[width=7cm]{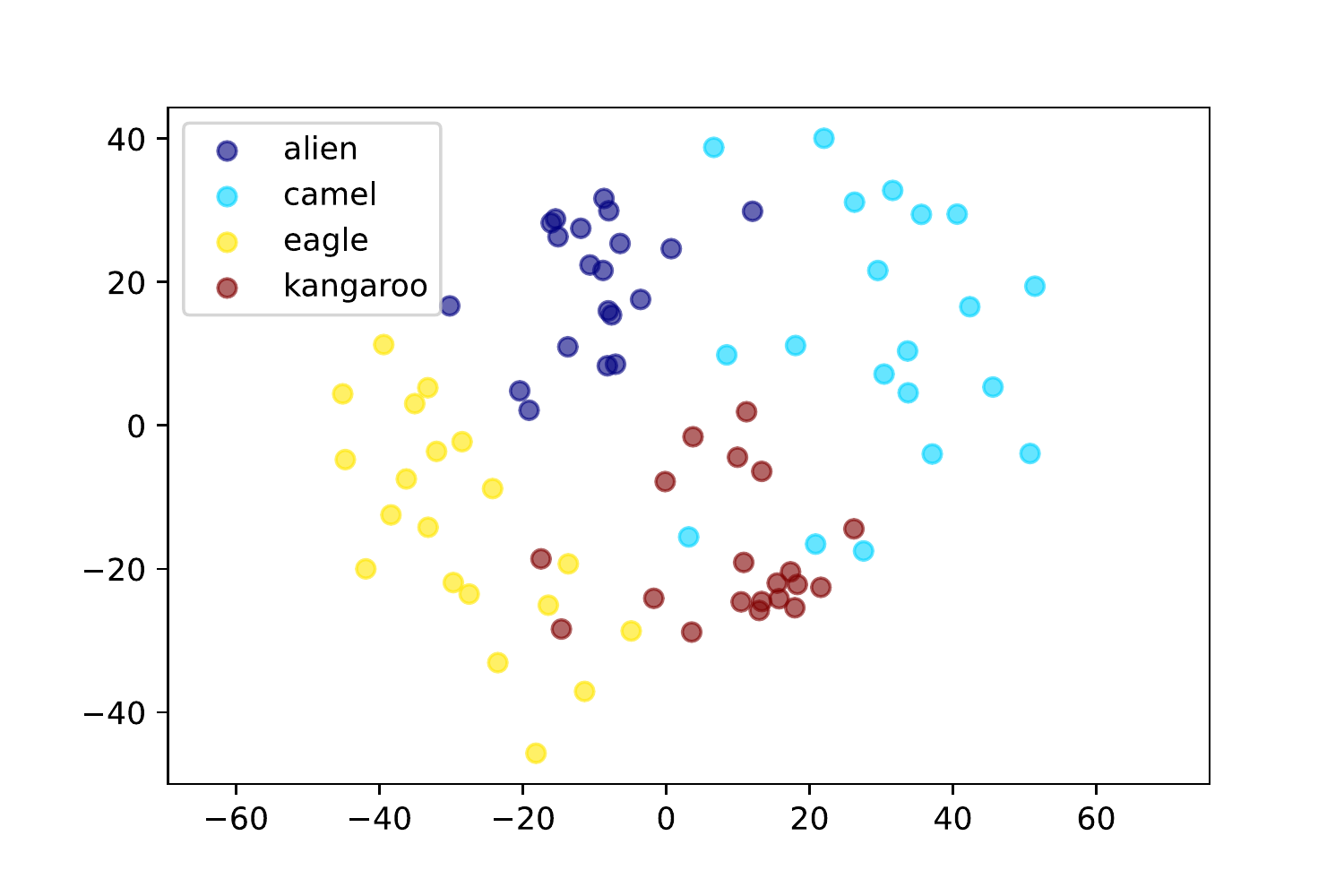} }}
    \qquad
    \subfloat[20 Frames]{{\includegraphics[width=7cm]{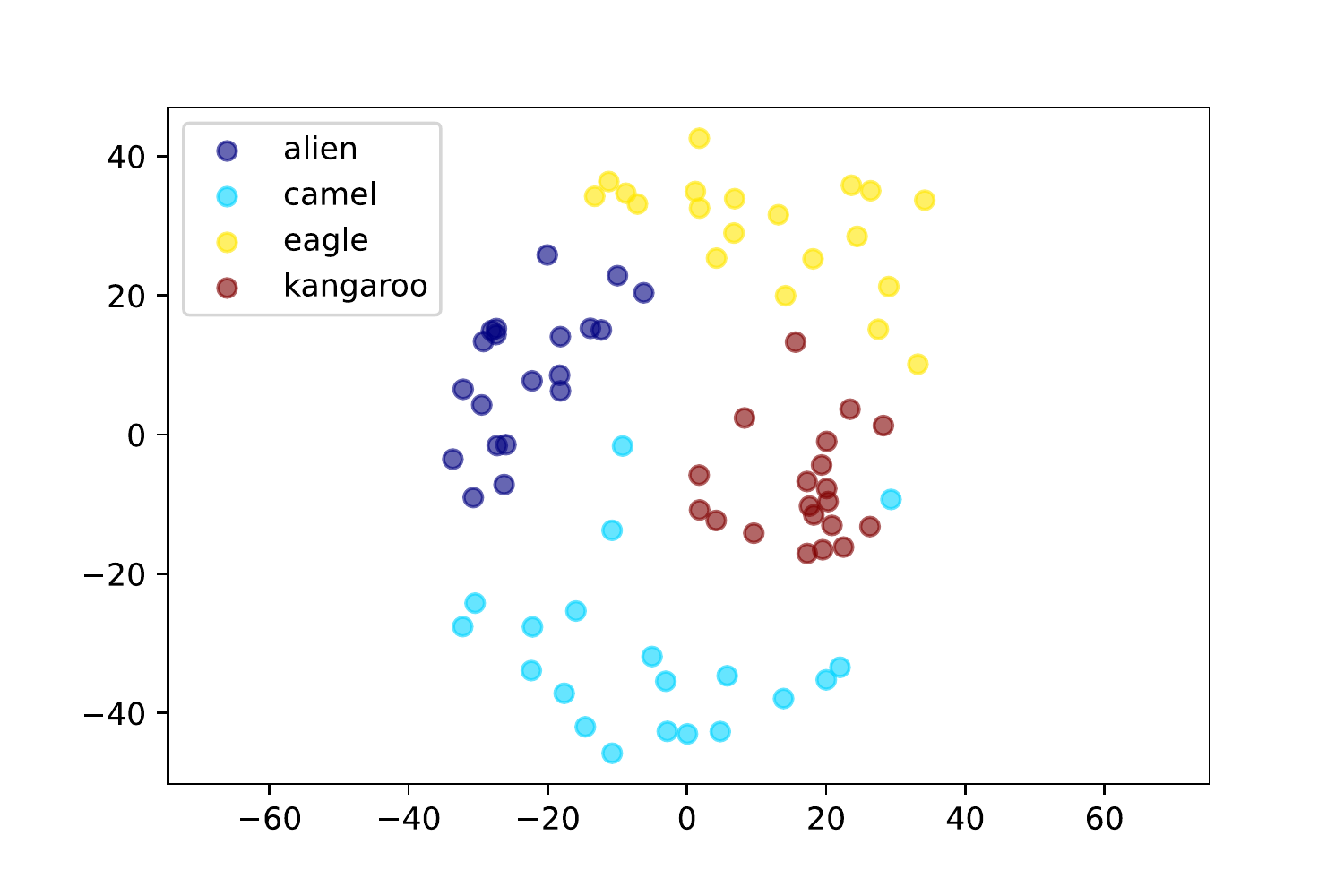} }}%
    \qquad
    \subfloat[50 Frames]{{\includegraphics[width=7cm]{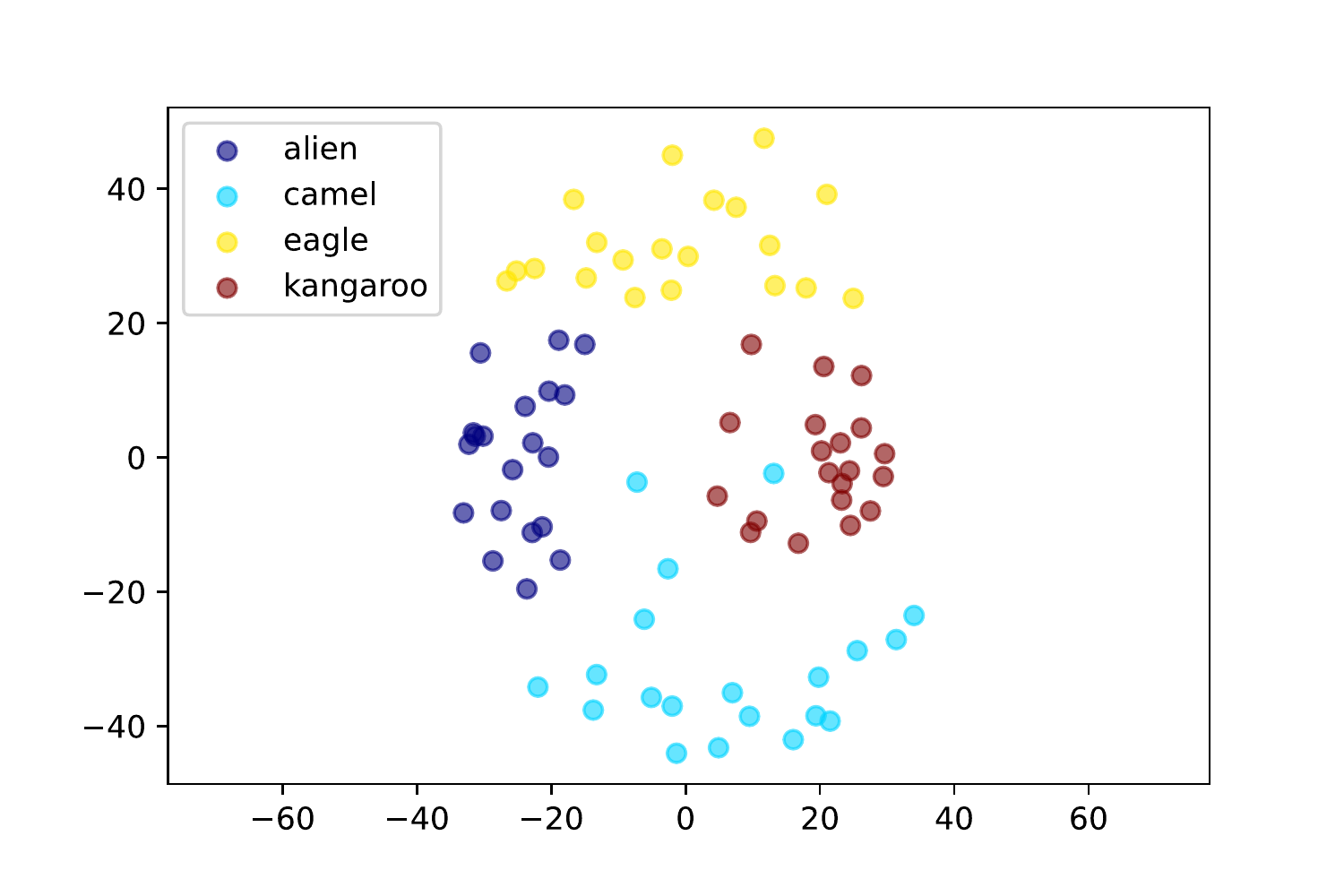} }}%
    \qquad
    \subfloat[100 Frames]{{\includegraphics[width=7cm]{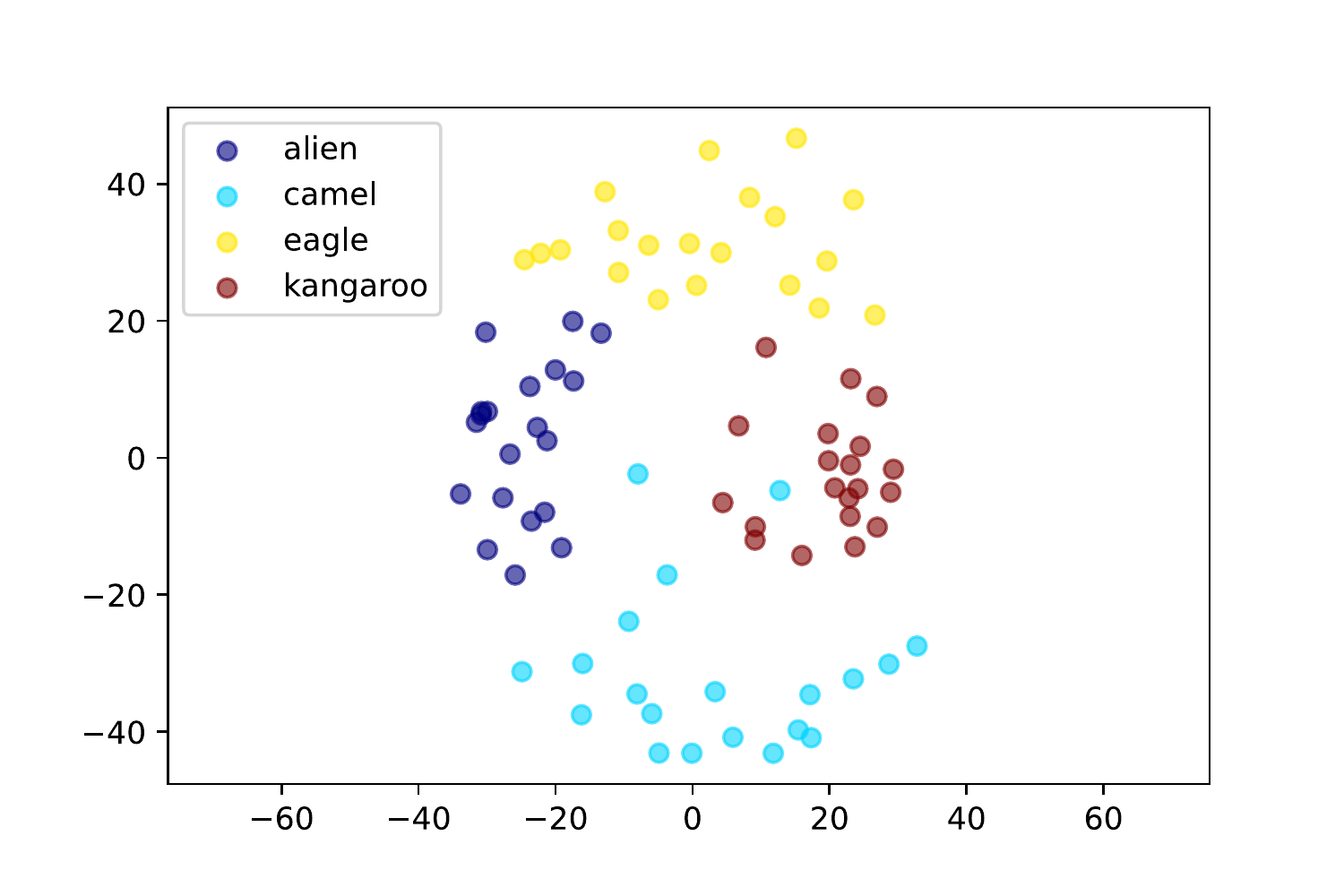} }}%
    \caption{MDS using ABD on 20 skeletons for each of four categories of images shown in Fig.~\ref{fig:binaryandskel}.}%
    \label{fig:frames}%
\end{figure}
Next we consider graphs of skeletonized binary images from the ShapeMatcher5 model dataset \cite{shapematcher}.  
ShapeMatcher is a program that can convert binary images to skeletons which can then be exported as graphs.
The graphs given by this program had a noisy structure that resulted merge trees with many small branches, which increased our runtime greatly. 
In order for our program to finish running in a reasonable amount of time we smoothed out the structures, resulting in the graphs shown in Figure \ref{fig:binaryandskel}b. 
We use ABD on 20 skeletons of 4 different images to compute pairwise distances at different numbers of frames.

In Figure \ref{fig:frames} we can see the results of MDS on four distance matrices with the same set of graphs but different matrix entries due to the number of frames used to compute each ABD. The only part of the distance matrix construction that is dependent on frames is the ABD calculation, so the runtime for the matrix construction increases linearly as the number of frames increases.

In all four plots we see similar graphs clustering together. 
When we compute ABD at 20, 50, and 100 frames, our clusters seem to be more clearly defined than when we only use 5 frames. 
We can also see that there is not much difference between scatterplots \ref{fig:frames}(b), \ref{fig:frames}(c), and \ref{fig:frames}(d). This suggests that using 20 frames yields a good estimate for the exact average branching distance between these skeletons. 

\section{Discussion}

In this paper, we have investigated the properties of the branching distance \cite{morozov}, showing that without inclusion of the infinite tail in the merge tree, the result is not a metric. 
The first interesting direction to go with this work is to see if the definition of the branching distance can be modified to take the tail into account. 
Will fixing the distance in this way propagate into a distance measure as defined on the rotated, embedded graphs? 
For the average branching distance, we would also be interested to see if, rather than working with finitely many directions, defining the distance to be some sort of integral over the branching distance in all directions would give stronger metric properties. 

The next direction for future work is to improve the computational techniques used in our code.
Our implementation makes some trade-offs between runtime and accuracy. 
First, for a fixed direction, we implement a binary search for branching distance between the resulting merge trees which leaves behind a small error. 
While in theory this could be iterated enough to find an exact distance (since only one possible value would remain), that approach was not practical.
In order to prevent an extremely large number of iterations, we choose a threshold for this error as a stopping criterion, which worked well enough in practice to give meaningful differentiation.

The second tradeoff between runtime and accuracy is caused by our choice to compute a specified set of angles rather than attempting to achieve an accurate representation of the entire range of angles $[0,2\pi)$. 
We cannot test the entire interval as there would be an unbounded number of frames, and a mathematically accurate representation would require a determination of the cheapest rooted branching pairs and how they change due to rotations for the entire interval.
It is possible that there is mathematical justification that can be undertaken for limiting the necessary number of directions, and indeed there is prior work on finding a rotation which would catch certain kinds of features in the data~\cite{micka2020} under various topological signatures (like the merge tree), although it mainly attempts to find a single direction to catch certain features. 
However, we leave this for future work as it is well beyond the scope of this project. 

Another issue is that of normalization.  We currently shift so that average vertex value is 0, which results in several instability issues.
First off, this procedure is not immune to the addition of regular vertices and as such is not uniform when viewing the merge trees as a topological space. 
Second, we normalize each pair of merge trees at every direction choice, thus potentially leading to wildly different normalization shifts from one rotation to the next. 
In future work, we would seek to remedy these issues potentially by mean-shifting the original shapes. 

Additionally, for this paper we have chosen to to simultaneously rotate the graphs rather than attempt to compare each orientation of one to each orientation of the other. 
Note that rotating the graphs simultaneously assumes some sort of ideal starting alignment between the two graphs. For example, it would not make sense to simultaneously rotate two different orientations of the same graph.
We also considered rotating only one graph, but concluded that this method would be inaccurate as it could consider irrelevant alignments.
For example, two identical objects will seem very different if only one is being rotated.
A consideration of all pairwise orientation matchings would increase the number of branching distance computations from $n$ to $n^2$ where $n$ is the number of specified frames. 
While this method may highlight different aspects of the compared graphs and, in particular, make the metric orientation invariant, it is too computationally expensive for our current implementation to handle.
However, it is possible that this modification would be of interest in applications such as shape comparison.

Finally, the data sets we considered involved only connected graphs, but in situations where we run into disconnected graphs, our implementation restricts to just the largest connected component. 
This of course will reduce the accuracy in practice, since parts of the data are  ignored.
Implementing an algorithm that alters \isepssimilar to account for disconnected components could improve our distance function. 
To do this, we would also have to define a new cost that basically accounts for deleting or adding an edge between two vertices, such that the two vertices that were connected in one graph match two vertices that were not connected in the other graph, similar to edit distance calculations between graphs.  This approach seems computationally prohibitive, but would likely lead to better comparisons if an approach can be made efficient.

\section{Acknowledgements}

This paper represents work done during Michigan State University's SURIEM REU program in Summer 2020. 
Due to the ongoing COVID-19 pandemic, all work was done remotely. 
We are grateful to the anonymous reviewer who pointed out the issue with the merge tree tails in the definition of the branching distance.
The funding for this project was supported by the National Science Foundation (NSF Award No. 1852066), the National Security Agency (NSA Grant No. H98230-20-1-0006), and Michigan State University.
The work of Erin Chambers was supported in part by NSF grants CCF-1614562, CCF-1907612, and DBI-1759807. 
The work of Elizabeth Munch was supported in part by NSF grants CCF-1907591, CCF-2106578, CCF-2142713, and  DEB-1904267.

We used the tools made available by \cite{csacademy, graphonline, shapematcher} to draw and export graph data that could be read into Python as NetworkX Graph objects. 
To examine the functionality of our code, we utilized graph data from several sources, including \href{https://www.openstreetmap.org/export}{OpenStreetMap.org}, the \href{https://earthworks.stanford.edu/}{Stanford EarthWorks Library}, and the Map Construction \cite{mapconstruction} project. %
These data do not appear in this paper but are included in our GitHub repository.

\printbibliography
\end{document}